\documentclass[11pt]{article}
\usepackage{amssymb}
\usepackage{amsmath}
\usepackage{amsthm, multirow}
\usepackage{graphicx, authblk}
\usepackage{algorithm}
\usepackage{algpseudocode}
\usepackage[round]{natbib}
\usepackage{fullpage}
\usepackage{xspace}

\usepackage{mathtools}

\algnewcommand\algorithmicinput{\textbf{INPUT:}}
\algnewcommand\INPUT{\item[\algorithmicinput]}
\algnewcommand\algorithmicoutput{\textbf{OUTPUT:}}
\algnewcommand\OUTPUT{\item[\algorithmicoutput]}
\allowdisplaybreaks
\usepackage{hyperref}[]
\hypersetup{
    colorlinks=true,
    linkcolor=blue,
    filecolor=magenta,      
    urlcolor=blue,
      citecolor=blue,
    }

\usepackage{cleveref}

\newcommand{\ou}{\"{o}}
\newtheorem{lemma}{Lemma}
\newtheorem{proposition}{Proposition}
\newtheorem{corollary}{Corollary}

\newtheorem{theorem}{Theorem}
\newtheorem{assume}{Assumption}
\newtheorem{model}{Model}

\newcommand{\p}{\mathcal P}
\newcommand{\hatp}{\widehat {\mathcal P}}

\DeclareMathOperator*{\argmin}{argmin}
\newcommand{\sepp}{\textsc{sepp}\xspace}
\newcommand{\pdp}{\textsc{pdp}\xspace}
\newcommand{\sepps}{\textsc{sepp}s\xspace}
\DeclareMathOperator*{\poiss}{\textsc{Poisson}}

\date{\vspace{-5ex}}

\title{Detecting Abrupt Changes in \\High-Dimensional Self-Exciting Poisson Processes}
\author[1]{Daren Wang}
\author[2]{Yi Yu}
\author[1]{Rebecca Willett}
\affil[1]{Department of Statistics, University of Chicago}
\affil[2]{Department of Statistics, University of Warwick}

\begin{document}
\maketitle

\begin{abstract}
High-dimensional self-exciting point processes have been widely used in many application areas to model discrete event data in which past and current events affect the likelihood of future events.  In this paper, we are concerned with detecting abrupt changes of the coefficient matrices in discrete-time high-dimensional self-exciting Poisson processes, which have yet to be studied in the existing literature due to both theoretical and computational challenges rooted in the non-stationary and high-dimensional nature of the underlying process.  We propose a penalized dynamic programming approach which is supported by a theoretical rate analysis and numerical evidence. 

\end{abstract}

\section{Introduction}
Self-exciting point processes (\sepps) are useful in modelling many types of discrete event data in which past and current event help determine the likelihood of future events.  Such data are ubiquitous in application areas including crime science \citep[e.g.][]{egesdal2010statistical}, national security \citep[e.g.][]{lewis2012self}, finance \citep[e.g.][]{chavez2012high} and neuroscience \citep[e.g.][]{linderman2016bayesian}, to name but a few.

\sepps were, arguably, first rigorously studied in a mathematical framework by \cite{hawkes1971spectra}, where the eponymous Hawkes process was proposed.  Since the debut of the Hawkes process, there have been tremendous efforts poured into different aspects of understanding and utilizing the univariate Hawkes process; see \cite{laub2015hawkes} and \cite{hawkesreview} for comprehensive and contemporary reviews.  More recently, due to the availability of richer datasets and computational resources, attention has shifted to multivariate and even high-dimensional \sepps, where different coordinates might correspond to different geographic locations, different neurons in a biological neural network, people in a social network, etc.  See, for instance, \cite{hall2016inference}, \cite{mark2018network}, \cite{chavez2012high} and \cite{ertekin2015reactive}.  
  
In these high-dimensional settings, understanding how events in one coordinate influence the likelihood of events in another coordinate provides valuable insight into the underlying process.  We call the collection of these influences between pairs of coordinates a ``network'', and this paper describes novel methods for detecting abrupt changes in this network with theoretical performance bounds that characterize the accuracy of the change point estimation and how strong the signals must be to ensure reliable estimation.

While change point detection has a long and rich history, we are unaware of any preexisting change point methodology that can be used to detect changes in \sepps in high dimensions.  Some recent high-dimensional change point detection work is briefly discussed as follows.  \cite{wang2018optimal} and \cite{padilla2019change} studied the change point detection in Bernoulli networks and dynamic random dot product graphs, respectively.  \cite{cho2015multiple}, \cite{cho2016change}, \cite{matteson2014nonparametric}, \cite{wang2018high}, \cite{dette2018relevant} and others investigated high-dimensional mean change problems.  \cite{wang2017optimal}, \cite{aue2009break} and others were concerned with high/multi-dimensional covariance structure changes.  \cite{safikhani2017joint} and \cite{wang2019localizing} exploited the high-dimensional vector autoregressive models and provided change point detection results thereof.   \cite{li2017detecting} focused on a low-dimensional Hawkes process setting in which the processes may be characterized by a small number of parameters.  

The lack of results on the change point analysis on high-dimensional \sepps can be related to its nonlinearities inherent to the model.  Note that, \sepps can be viewed as a nonlinear autoregressive process, and even detecting changes for linear vector autoregressive processes is an active area of investigation \citep{wang2019localizing}. The nonlinearities associated with \sepps further complicate the change point detection problem.  

This paper describes a computationally- and statistically-efficient methodology for detecting changes in the network underlying \sepps. At the heart of our method lies a penalized dynamic programming algorithm that estimates the times at which each change occurs when the underlying network is sparse, i.e.~when the number of network edges is small relative to the number of pairs of network nodes.  In this paper, we also apply our method to neuron spike train data sets to help pinpoint the times at which the functional networks might change due to the changes of the state of consciousness.


\subsection{Problem formulation}

The detailed model considered in this paper is introduced as follows.

\begin{model} \label{model:sepp change point}
Let $\{X(t)\}_{t=1}^T \subset \mathbb{Z}^M$ be a discrete-time Poisson process.  For each $t \in \{1, \ldots, T\}$, let $\mathcal{X}(t) = (X(1),\ldots, X(t)) \in \mathbb{R}^{M \times t}$ consist all the history up to time $t$.  For each $t \in \{1, \ldots, T-1\}$ and $m \in \{1, \ldots, M\}$, suppose that given $\mathcal{X}(t)$, all coordinates $\{X_m(t+1)\}$ are conditionally independent and the conditional distribution of $X_m(t+1)$ is a Poisson distribution, i.e.
	\begin{align}\label{eq:change point SEPP}
		X_m(t+1) | \mathcal X(t) \sim \poiss (\exp\{\lambda_m (t)\}),
	\end{align}
	where
	\begin{equation}\label{eq-lambda-definition}
		\lambda_m (t) = v + A_{m}^*(t)  g_t\{\mathcal X( t)\},		
	\end{equation}
	the matrix $A ^* (t)\in \mathbb R^{M\times M}$ is the coefficient matrix at time point t, $A^*_m(t)$ is the $m$-th row of $A^*(t)$ and $g_t(\cdot):\, \mathbb{R}^{M \times t} \to \mathbb{R}^M$ is a $M$-dimensional vector-valued function.

Suppose that there exists an integer $K \geq 0$ and time points $\{\eta_k\}_{k = 0}^{K+1}$, called change points, satisfying $1 = \eta_0 < \eta_1 < \ldots < \eta_K \leq T < \eta_{K+1} = T+1$ and
	\[
		A^*(t) \neq A^*(t - 1) \quad \mbox{if and only if} \quad t \in \{\eta_k\}_{k = 1}^{K}.
	\]
	Let the minimal spacing and the minimal jump size be defined as 
	\[
		\Delta = \min_{k = 1, \ldots, K + 1} (\eta_k - \eta_{k-1}) \quad \text{and} \quad \kappa = \min_{k = 1, \ldots, K + 1} \|A^*(\eta_k) - A^*(\eta_{k-1})\|_{\mathrm{F}},
	\]
	respectively, where $\|\cdot\|_{\mathrm{F}}$ denotes the Frobenius norm of a matrix.	
\end{model}

It is worth mentioning that $\{X(t)\}_{t = 1}^T$ defined in \Cref{model:sepp change point} is an \textsc{sepp}, where each $X_m(t)$ is conditionally distributed as a Poisson random variable.  We therefore refer to \eqref{eq:change point SEPP} a self-exciting Poisson processes.  In the sequel, when there is no ambiguity, we will also refer self-exciting Poisson processes as \sepps.

In fact, \Cref{model:sepp change point} is a generalization of a stationary \sepp process, which assumes that the coefficient matrices $A^*(t) = A^*(1)$, $t \in \{1, \ldots, T\}$.  Stationary \sepp models have been well-studied in the existing literature, including \cite{hall2018learning} and \cite{mark2018network}, where it has been shown that the coefficient matrix of the point process can be estimated by an $\ell_1$-penalized likelihood estimator.  

Given $\{X(t)\}_{t=1}^T$ satisfying \Cref{model:sepp change point}, our main task is to estimate $\{\eta_k\}_{k = 1}^K$ accurately.  To be specific, we seek estimators $\{\widehat{\eta}_k\}_{k = 1}^{\widehat{K}}$ such that as the sample size $T \to \infty$, with probability tending to 1, it holds that
	\begin{equation}\label{eq-cons-def}
		\widehat{K} = K \quad \mbox{and} \quad \frac{\epsilon}{\Delta} = \frac{\max_{k = 1, \ldots, K}|\widehat{\eta}_k - \eta_k|}{\Delta} \to 0.
	\end{equation}
	For the change point estimators satisfying \eqref{eq-cons-def}, we call them \emph{consistent} change point estimators.  We will also call $\epsilon$ the \emph{localization error}.
	
To the best of our knowledge, we are the first to study the high-dimensional \sepps with change points.  In addition to the mathematical introduction of the model, we investigate the consistency of the abrupt change point location estimators, under minimal conditions.  The proposed penalized dynamic programming approach in \Cref{section:main} is computationally efficient and tailored for this novel setting.

\subsection*{Notation}
For any integer pair $(t_1, t_2) \in \mathbb{Z}^2$, let $[t_1, t_2]$ denote the integer interval $[t_1, t_2] \cap \mathbb{Z}$.  Same notation applies to open intervals.  For any matrix $A \in \mathbb R^{M\times M}$, let $A_{m}$ denote the $m$th row of $A$ and $A_{m, m'}$ denote the $(m, m')$th entry of $A$.  With some abuse of notation, for any vector $v$ and any matrix $M$, let $\|v\|_2$, $\|v\|_1$, $\|M\|_{\mathrm{F}}$ and $\|M\|_1$ be the $\ell_2$- and $\ell_1$-norms of $v$, the Frobenius norm of $M$ and the $\ell_1$-norm of $\mathrm{vec}(M)$, respectively, where $\mathrm{vec}(M)$ is the vectorized version of $M$ by stacking all the columns of $M$.  For any  $v(t):\, [1, T] \to \mathbb R^m$, let 
	\[
		\| Dv\|_0 = \sum_{t  =2}^T I \{v(t-1) \neq v(t)\},
	\]
	where $I\{ \cdot\} \in \{0, 1\}$ is the indicator function.  For any set $S \subset \{(m, m'): \, m ,m' = 1, \ldots, M\}$, let $A_S  \in \mathbb R^{M \times M}$ satisfy
	$$ (A_S)_{m, m'} = \begin{cases}
 	A_{m, m'}, & (m, m') \in S, \\
 	0, & \text{otherwise.}
	\end{cases} $$	
	Given any $A(t): \, [1, T] \to \mathbb R^{M\times M}$ and any $J \subset [1, T]$, if $A(\cdot)$ is unchanged in $I$, then we denote $A(I) = A(t)$, $t \in I$.  

\section{The Penalized Dynamic Programming Algorithm}
\label{section:main}

To detect the change points in \Cref{model:sepp change point}, we propose the penalized dynamic programming (\pdp) algorithm, which is stated in \eqref{eq:SEPP DP} with necessary notation in \eqref{eq:definition of A hat}, \eqref{eq:penalized likelihood} and \eqref{eq:constrain set}.  The \pdp consists of two layers: estimation of the coefficient matrices $A^*(t)$, $t \in [1, T]$, and estimation of the change points.  

For the coefficient matrix estimation, we let $\widehat{A}(I)$ be the penalized log-likelihood estimator of the coefficient matrix over an integer interval $I \subset [1, T]$, i.e.
	\begin{align} \label{eq:definition of A hat}
		\widehat A(I) = \argmin_{A \in \mathcal C} H(A, I),
	\end{align}
	where $H(A, I)$ and $\mathcal C$ are the penalized log-likelihood function and the constrained domain of the coefficient matrices, respectively.  To be specific, with a pre-specified tuning parameter $\lambda > 0$ and $I = [s, e]$, let
	\begin{align}\label{eq:penalized likelihood} 
		H(A, I) = \sum_{t = s}^{e - 1} \sum_{m=1}^M \left(\exp\left[v + A_m g_t\{\mathcal X (t)\}\right] - X_m (t+1) [v + A_m g_t\{\mathcal X (t)\}] + \lambda |I |^{1/2} \|A\|_1\right)
	\end{align}
	and
	\begin{align} \label{eq:constrain set}
		\mathcal C = \left \{A \in \mathbb R^{M \times M }:\, \max_{m = 1, \ldots, M} \|A_{m}\|_1 \le 1 \right \}.
	\end{align}

The loss function $H(\cdot, \cdot)$ is a penalized logarithmic conditional likelihood function, recalling that $X_m(t+1)$ given $\mathcal{X}(t)$ follows a Poisson distribution with intensity $\exp\left[v + A_m g_t\{\mathcal X (t)\}\right]$.  The penalty term $\lambda |I|^{1/2}$ in \eqref{eq:penalized likelihood} is introduced in a way such that the tuning parameter $\lambda$ is independent of the interval length.  The term $|I|^{1/2}$ reflects the order of the standard error of the sum of $|I|$ marginal log-likelihood functions.  We elaborate on this scaling factor and its derivation in \Cref{lemma:Azuma Hoeffding} and its proof. 

The constraint on $\mathcal{C}$ is to ensure that the \sepp process as vector-valued time series is stable \citep[see e.g.][]{lutkepohl2005new}.  As for stationary \sepp estimation, \cite{mark2018network} proposed a constraint similar to \eqref{eq:constrain set}.
	
\medskip
Given the above framework, we can now consider estimating change points by setting
	\begin{align} \label{eq:SEPP DP}
		\widehat { \mathcal P }  = \argmin_{\mathcal P} \left\{\sum_{ I \in \mathcal P }   H( \widehat A(I), I) + \gamma | \mathcal P|\right\},
	\end{align}
where $\gamma > 0$ is a tuning parameter, the minimization is over all possible interval partitions of $[1, T]$  and $\mathcal P$ denotes one such partition.  To be specific, an interval partition has the form $\mathcal P = \{I_k, \, k=1, \ldots, K_{\mathcal{P}}\}$ and satisfies $I_{k'} \cap I_{k} = \emptyset$ and $\bigcup_{k=1}^{K_{\mathcal{P}}} I_k = [1, T]$.  Once $\widehat{\mathcal{P}}$ is at hand, we let $\widehat K = |\widehat {\cal P}| - 1 \geq 0$, $\eta_{\widehat{K} +  1} = T+1$ and 
	\[
		\widehat{\mathcal{P}} = \left\{\{1, \ldots, \widehat{\eta}_1 - 1\},\, \{\widehat{\eta}_k, \ldots, \widehat{\eta}_{k+1} - 1\}_{k = 1}^{\widehat{K}}\right\}.
	\] 
	We call $\{\widehat{\eta}_k\}_{k = 1}^{\widehat{K}}$ the change point estimators induced by  $\widehat{\mathcal{P}}$.

The optimization problem in \eqref{eq:SEPP DP} is known as the minimal partition problem on a linear chain graph and can be solved using dynamic programming \citep[e.g.][]{friedrich2008complexity} with the worst case computational cost of order $O\{T^2 \mathrm{Cost}(T)\}$, where $\mathrm{Cost}(t)$ denotes, in our case, the computational cost of computing $\widehat A(I)$ in the interval $I$ with $|I| = t$.  We remark that there has been a line of attack on the computational aspect of optimizing the minimal partition problem, including \cite{killick2012optimal} and \cite{maidstone2017optimal}, among others. 

For completeness, we summarize the \pdp procedure in \Cref{algorithm:PDP} below.  The quantities and functions involved there are defined in \eqref{eq:definition of A hat}, \eqref{eq:penalized likelihood} and \eqref{eq:constrain set}.
\begin{algorithm}[htbp]
\begin{algorithmic}
	\INPUT Data $\{X(t)\}_{t=1}^{T}$, tuning parameters $\lambda, \gamma > 0$.
	\State $(\mathcal{B}, s, t, \mathrm{FLAG}) \leftarrow (\emptyset, 0, 2, 0)$
	\While{$s < T-1$}
		\State $s \leftarrow s + 1$
		\While{$t < n$ and $\mathrm{FLAG} = 0$}
			\State $t \leftarrow t+1$ 			
			\If{$\min_{l = s+1, \ldots, t-1}\{H(\widehat{A}([s, l]), [s, l]) + H(\widehat{A}([l+1, t]), [l+1, t]) + \gamma < H(\widehat{A}([s, t]), [s, t])\}$} 
				\State 
				\begin{align*}
					& s \leftarrow \min \argmin_{l = s+1, \ldots, t-1}\{H(\widehat{A}([s, l]), [s, l]) + H(\widehat{A}([l+1, t]), \\
					& \hspace{5cm} [l+1, t]) + \gamma < H(\widehat{A}([s, t]), [s, t])\}
				\end{align*}
				\State $\mathcal{B} \leftarrow \mathcal{B} \cup \{s\}$
				\State $\mathrm{FLAG} \leftarrow 1$
			\EndIf
		\EndWhile
	\EndWhile
	\OUTPUT The set of estimated change points $\mathcal{B}$.
\caption{Penalized dynamic programming. PDP$(\{X(t)\}_{t=1}^{n}, \lambda, \gamma)$}
\label{algorithm:PDP}
\end{algorithmic}
\end{algorithm}

\subsection{The localization rate of the Penalized Dynamic Programming estimators}

In order to establish the consistency of the change point estimators resulting from the \pdp procedure detailed in \Cref{algorithm:PDP}, we first impose \Cref{assume:parameters}.  

\begin{assume}\label{assume:parameters} 
Let $\{X(t)\}_{t=1}^T \subset \mathbb{Z}^M$ be a discrete-time \sepp generated according to \Cref{model:sepp change point} and satisfying the following.
\begin{itemize}
\item [{\bf A1.}] There exists a subset $S \subset \{(m, m'): m, m' = 1, \ldots, M\}$ such that for all $t \in [1, T]$, $A_{m, m'}^*(t) = 0$, if $(m, m') \notin S$.  Let $d = |S|$. 
\item [{\bf A2.}] It holds that
	\[
		\max_{t = 1, \ldots, T} \max_{m =1, \ldots, M} \|A^*_{m}(t)\|_1 \le 1.
	\]
\item [{\bf A3.}] For any $\xi > 0$, there exist absolute constants $C_{\Delta, 1}, C_{\Delta, 2} > 0$ such that
	\[
		\Delta \ge C_{\Delta, 1} T \quad \mbox{and} \quad \Delta \ge C_{\Delta, 2} \log^{2+\xi}(TM) d ^2\max\{\kappa^{-2},\, \kappa^{-4}\}.
	\] 
\item [{\bf A4.}] There exist absolute constants $p \in \mathbb Z^+$ and $\omega > 0$ such that for any $t$, the matrix 
	\[
		\mathbb{E}[g_t\{\mathcal X(t)\}g_t\{\mathcal X(t)\}^\top |\mathcal X (t - p)] -\omega I_{M}
	\] 
	is positive definite, where $I_M \in \mathbb{R}^{M \times M}$ is an identity matrix.  In addition, $v$ and $\|g_t(\cdot)\|_{\infty}$, for all $t$, are uniformly upper bounded by an absolute constant $C_g > 0$.  
\end{itemize}
\end{assume}

\Cref{model:sepp change point} and \Cref{assume:parameters} completely charaterize the problem with model parameters $M$ (the dimensionality of the time series), $d$ (the sparsity parameter indicating an upper bound of the number of nonzero entries in all the coefficient matrices), $\Delta$ (the minimal spacing between change points), and $\kappa$ (the minimal jump size), along with the sample size $T$.  The consistency we are to establish is based on allowing $M$ and $d$ to diverge and $\kappa$ to vanish as the sample size $T$ diverges unbounded.  

The number of parameters at each time point is of order $M^2$, which is allowed to well exceed the sample size.  A sparsity constraint therefore comes into force in Assumption {\bf A1}, which is a standard assumption in the high-dimensional statistics literature.  Note that the set $S$ is the union of all $(m,m')$ pairs with a nonzero entry in any coefficient matrix.  Assumption {\bf A2} echoes the imposition of the constraint domain $\mathcal{C}$ \eqref{eq:constrain set} in the optimization \eqref{eq:definition of A hat}, to ensure the stationarity of the \sepp.  In fact, the constant one in the upper bound can be relaxed to any absolute constant and is set to be one in this paper for identification issue.  To be specific, what goes into the model is the product of $A_m(t)$ and $g_t\{\mathcal{X}(t)\}$, and the latter is assumed to be upper bounded in sup-norm in  Assumption {\bf A4}.

Assumption {\bf A3} can be regarded as a signal-to-noise assumption.  It is required that the minimal spacing $\Delta$ is at least of a constant fraction of the total sample size, which implies the number of change points is of order $O(1)$.  This might appear to be strong compared to other change point detection literature, however, the problem we are facing here is challenging due to the nonlinearity of the \sepp model.  In fact, Assumption {\bf A3} is a mild condition and covers some challenging scenarios.  For instance, Assumption {\bf A3} holds if $M \asymp \exp(T^{1/2})$, $d \asymp T^{1/4}$ and $\kappa \asymp \log(T)$.  The quantity $\xi$ can be set arbitrarily small and it ensures the consistency of the estimator which will be explained after \Cref{theorem:SEPP change point}.
 
Assumption {\bf A4} can be interpreted as the restricted eigenvalue condition for \sepp processes.  We refer readers to Section~4 of \cite{mark2018network} for a number of common self-excited point process models satisfying Assumption {\bf A4}.

\medskip

In what follows, we show the consistency of \pdp in \Cref{theorem:SEPP change point}.
	 
\begin{theorem} \label{theorem:SEPP change point}
Let $\{X(t)\}_{t=1}^T \subset \mathbb{Z}^M$ be an \sepps generated from \Cref{model:sepp change point} and satisfying \Cref{assume:parameters}.  Let $\{\widehat \eta_k\}_{k=1}^{\widehat K}$ be the change point estimators from the \pdp algorithm detailed in \Cref{algorithm:PDP} with tuning parameters
	\begin{equation}\label{eq-tuning-order}
		\lambda =  C_\lambda \log(TM) \quad \text{and} \quad \gamma =  C_\gamma \log^2(TM) d\left(1 +  d\kappa^{-2}\right),
	\end{equation} 
	where $C_\lambda, C_\gamma > 0$ are absolute constants, depending only on $p$, $\omega$, $C_{\Delta, 1}$, $C_{\Delta, 2}$ and $C_g$.  We have that
	\[
		\mathbb{P}\left\{\widehat{K} = K \quad \mbox{and} \quad \max_{k = 1, \ldots, K}|\widehat{\eta}_k - \eta_k| \le C_{\epsilon}d^2 \log^2(TM) \max\{\kappa^{-2}, \kappa^{-4}\}\right\} \geq 1- 2 (TM)^{-1},
	\]
	where $C_\epsilon > 0$ is an absolute constant only depending on $p$, $\omega$, $C_{\Delta, 1}$, $C_{\Delta, 2}$ and $C_g$.
\end{theorem}

The proof of \Cref{theorem:SEPP change point} is deferred to \Cref{sec-proof-thm}, where it can be seen that the order of the estimation error is of the form 
	\[
		\frac{\lambda^2 d}{\kappa^2} + \frac{\lambda^2 d^2}{\kappa^4} +  \frac{\gamma}{\kappa^2}.
	\]
	Due to the signal-to-noise ratio condition in Assumption~\textbf{A3}, we have that
	\begin{align*}
		\frac{\max_{k = 1, \ldots, K}|\widehat{\eta}_k - \eta_k|}{\Delta} \lesssim \frac{d^2 \log^2(TM) \max\{\kappa^{-2}, \kappa^{-4}\}}{\Delta} \lesssim \frac{d^2 \log^2(TM) \max\{\kappa^{-2}, \kappa^{-4}\}}{d^2 \log^{2 + \xi}(TM) \max\{\kappa^{-2}, \kappa^{-4}\}} \to 0,
	\end{align*}
	as $T \to \infty$.  This explains the role of the quantity $\xi$ in Assumption~\textbf{A3} and shows the consistency of the \pdp algorithm.  In fact, if we let $d = 1$ and assume $\kappa > 1$, then the localization error we derived here coincides with the optimal localization error in the univariate mean change point detection problem \citep[e.g.][]{wang2020univariate}.

Two tuning parameters are involved, where $\lambda$ is used in the optimization \eqref{eq:penalized likelihood} to recover the sparsity in estimating high-dimensional coefficient matrices, and $\gamma$ is involved in optimizing \eqref{eq:SEPP DP} to penalize the over-partitioning.  The order of $\lambda$ required in \eqref{eq-tuning-order} is a logarithmic quantity in $T$ and $M$, which is resulted from a union bound argument applied to a sub-exponential concentration bound.  The requirement on $\gamma$ is essentially that $\gamma \asymp \lambda^2\left(d + d^2\kappa^{-2}\right)$, which can be intuitively explained as an upper bound on the difference between $H(\widehat{A}(I_1), I_1) + H(\widehat{A}(I_2), I_2)$ and $H(\widehat{A}(I_1 \cup I_2), I_1 \cup I_2)$, where $I_1$ and $I_2$ are two relatively long, non-overlapping and adjacent intervals, and there is no true change point near the shared endpoint of $I_1$ and $I_2$.  In this case, one would not wish to partition $I_1 \cup I_2$ into $I_1$ and $I_2$.  If we only focus on the log-likelihood functions, over-estimating will result in that 
	\[
		H(\widehat{A}(I_1), I_1) + H(\widehat{A}(I_2), I_2) < H(\widehat{A}(I_1 \cup I_2), I_1 \cup I_2).
	\]
	The penalty we impose through $\gamma$ will therefore avoid this over-partitioning.


\subsection{Comparisons with related work}
\label{sec:prior}
In a broad sense, as we have mentioned, there have been numerous  existing papers on different aspects of \sepps.  In fact, another related area is the analysis of piecewise-stationary time series models, where we also see a vast volume of existing papers.  The two most related papers are \cite{mark2018network}, which is concerned with a stationary, high-dimensional \sepp process, and \cite{wang2019localizing}, which studies a piecewise-stationary high-dimensional linear process. 



\cite{mark2018network} studied a stationary version of \Cref{model:sepp change point} with $K = 0$.  The penalized estimator of the coefficient matrix developed there is almost identical to the ones summoned in our problem in \eqref{eq:definition of A hat}.  There are a few fundamental differences between this paper and \cite{mark2018network}.  (1) Due to the piecewise-stationarity assumed in \Cref{model:sepp change point}, when estimating the coefficient matrices in \eqref{eq:definition of A hat} and \eqref{eq:penalized likelihood}, it is possible that there exists a true change point in the interval of interest and the estimator we seek is an estimator of a mixture of different true coefficient matrices.  (2) We provide a more refined analysis as an improved version of \cite{mark2018network}, for instance, the optimization constrain domain $\mathcal{C}$ defined in \eqref{eq:constrain set} is a cleaner version of its counterpart in \cite{mark2018network}; a subspace compatibility condition is required in \cite{mark2018network} to control the ratio of different norms of the coefficient matrix, and this assumption is shown to be redundant in our new analysis.

The other closest-related work is \cite{wang2019localizing}, where the change point localizing problem in the piecewise-stationary vector autoregressive models is investigated and a penalized dynamic programming approach was deployed there.  The main differences between this paper and \cite{wang2019localizing} comes from the underlying model.  The vector autoregressive model is a linear model in the sense that given $\mathcal{X}(t)$ the history data up till time point $t$, the conditional expectation of $X(t+1)$ is a linear combination of the columns of $\mathcal{X}(t)$, which is not the case here.  The self-exciting point process is a nonlinear model, and as we have mentioned, the logarithm of the conditional intensity is a linear function of the history.  Another key difference is that  \cite{wang2019localizing} are concerned with sub-Gaussian innovation sequences, while the counting processes we study here determine the heavy-tail properties of the data.

\section{Numerical Experiments}
\label{section:numeric}

In this section, we further examine the performances of the \pdp algorithm by numerical experiments, with simulated data analyzed in \Cref{sec-sim} and a real data set in \Cref{sec-real-data}.

\subsection{Simulated data analysis} \label{sec-sim}

We generate data according to \Cref{model:sepp change point} and \Cref{assume:parameters}.  In particular, we adopt the setting in  \cite{mark2018network} and assume that the design function $g_t(\cdot)$ is defined to be
	\begin{equation}\label{eq-g-specify}
		g_t\{\mathcal X(t)\} = \left(\min\{\mathcal{X}_1(t), \, C_g\}, \ldots, \min\{\mathcal{X}_M(t), \, C_g\}\right)^{\top} \in \mathbb{R}^M,
	\end{equation}
	where $C_g > 0$ is a constant, $\mathcal{X}(t)$ is an $M \times t$ matrix and $\mathcal{X}_m(t)$ denotes the $m$th row of $\mathcal{X}(t)$, $m \in \{1, \ldots, M\}$.   For the two tuning parameters $\lambda$ and $\gamma$ defined in \eqref{eq:penalized likelihood} and \eqref{eq:SEPP DP}, respectively, with the theoretical guidance in \Cref{theorem:SEPP change point}, we fix $\lambda =90 \log(TM)$ and $\gamma = \log^2(M)/2$ in all experiments in this section.

Since the piecewise-stationary \sepp model is first introduced here, we do not have direct competitors.  For illustration purpose, however, we compare our \pdp algorithm with the SBS-MVTS algorithm \citep{cho2015multiple} and E-Divisive procedure \citep{matteson2014nonparametric}, both of which are designed to detect abrupt change points in multivariate time series, but neither of which is designed specifically for the scenarios we are studying here.  Having said this, there are the reasons we choose these two competitors.  The SBS-MVTS can identify covariance changes in the high-dimensional autoregressive time series and the E-Divisive procedure can estimate of both the number and locations of change points under mild assumptions on the first or second moments of the underlying distributions.  Since Poisson random variables have the same means and variances, these two competitors may be able to detect the changes in Poisson processes with piecewise-constant parameters.  In all the simulated experiments, the tuning parameters for SBS-MVTS algorithm and E-Divisive procedure are selected according to the information-type criteria and permutation tests in the R \citep{R} packages wbs \citep{wbs-R} and ecp \citep{ecp-R}, respectively.

Let $\{\widehat \eta_k\}_{k=1}^{\widehat K}$ and $\{\eta_k\}_{k=1}^{K}$ be a collection of change point estimates and a collection of true change points, respectively.  We evaluate the estimators' performances by the absolute error $|K -\widehat K|$ and their Hausdorff distance.  The Hausdorff distance between two sets $\mathcal A$ and $\mathcal B$ is defined as
	\begin{equation}\label{eq-D-defineee}
		\mathcal D( \mathcal A, \mathcal  B ) = \max\{ d ( \mathcal A | \mathcal  B ),\, d(\mathcal  B | \mathcal  A)   \},
	\end{equation}
	where $$d(\mathcal  A | \mathcal  B) = \max_{a\in \mathcal A } \min_{b \in \mathcal B } |a-b|. $$  
 
In the sequel, we consider three settings.  Recall that $T$ is the total number of time points, $M$ is the dimensionality of the time series and $C_g$ is the threshold used in the design function $g_t(\cdot)$, which is specified in \eqref{eq-g-specify}.  Every setting is repeated 100 times.  Additional setting details are listed below. 
	\begin{itemize}
	\item [(a)]	One change point and varying jump size.  Fix $T = 300$, $M = 30$, $C_g = 6$ and the intercept $v = 1/2$, which is defined in \eqref{eq-lambda-definition}. Let
		\[
			A^*(t) = \begin{cases}
				 (\rho v_1, \rho v_2, 0_{M\times (M-2)}) \in \mathbb{R}^{M \times M}, & t \in [1, 150], \\
				 (\rho v_2, \rho v_1, 0_{M\times (M-2)}) \in \mathbb{R}^{M \times M}, & t \in [151, 300],
			\end{cases}
		\]
	where $v_1 \in \mathbb{R}^M$ with odd coordinates being 1 and even coordinates being $-1$, $v_2 = -v_1$, $0_{M \times (M-2)} \in \mathbb{R}^{M \times (M-2)}$ is an all zero matrix and $\rho \in \{0.15, 0.20, 0.25, 0.30, 0.35\}$.
	\item [(b)] Two change points and varying minimal spacing.  Let $T \in \{180, 240, 300, 360, 420\}$, $M = 40$, $C_g = 8$ and the intercept $v = 1/4$.  Let the coefficient matrices satisfy $(A^*(t))_{ij} = 0$, $|i - j| > 1$, $t \in [1, T]$,
		\[
			(A^*(t))_{ij} = \begin{cases}			
				\begin{cases}
 					0.15 & t \in [1, T/3] \cup (2T/3, T], \\
 					-0.15 & t \in (T/3, 2T/3],
 				\end{cases} & i = j, \\ \\
				\begin{cases}
	 				-0.15 & t \in [1, T/3], \\
 					0.15 & t \in (T/3, T],
 				\end{cases} & i - j = -1, \\ \\
				\begin{cases}
 					0.15 & t \in [1, 2T/3], \\
	 				-0.15 & t \in (2T/3, T],
 				\end{cases} & i - j = 1.
 			\end{cases}	
		\]
	\item [(c)]	Two change points and varying dimension.  Let $T = 450$, $C_g = 4$, $v =1/5$ and $M \in \{15, 20, 25, 30, 35\}$.  Let
		\[
			A(t) = \begin{cases}
				(v_1, v_2, v_3, 0_{M \times (M-3)}), & t \in [1, 150], \\
				(v_2, v_3, v_3, 0_{M \times (M-3)}), & t \in [151, 300], \\
				(v_3, v_2, v_1, 0_{M \times (M-3)}), & t \in [301, 450],
			\end{cases} 
		\]
		where $v_1, v_2, v_3 \in \mathbb{R}^M$ are
		\begin{align*}
			v_1 & = (-0.075, 0.15, 0.3, -0.3, 0, \ldots, 0)^{\top}, \\
			v_2 & = (\underbrace{0, \ldots, 0}_4, 0.375, -0.225, -0.075, 1.5, 0.225, 0, \ldots, 0)^{\top}, \\
			v_3 & = (\underbrace{0, \ldots, 0}_8, -0.15, -0.075, 0.45 , -0.225, 0, \ldots, 0)^{\top}.\end{align*} 
	\end{itemize}

\begin{table}[h]\small
\begin{center}
\begin{tabular}{ccccccc}
 	\multicolumn{7}{c}{Setting (a)} \\[5pt]
	Method & Metric & $\rho = 0.15$ & $\rho = 0.20$ & $\rho = 0.25$ & $\rho = 0.30$ & $\rho = 0.35$ \\ [3pt]
	\pdp & $\mathcal{D}$ & 3.1(9.8) & 1.1(1.0) & 0.7(0.5) & 0.6(0.5) & 0.6(0.5) \\
	SBS-MVTS &  & 282.6(69.1) & 226.5(119.9) & 114.7(130.8) & 47.3(52.9) & 9.3(21.3) \\
	E-Divisive &  & 151.0(0.0) & 151.0(0.0) & 151.0(0.0) & 151.0(0.0) & 151.0(0.0) \\ [3pt]
	\pdp & $|\widehat{K} - K|$ & 0.0(0.0) & 0.0(0.0) & 0.0(0.0) & 0.0(0.0) & 0.0(0.0) \\
	SBS-MVTS & & 0.9(0.2) & 0.7(0.4) & 0.4(0.5) & 0.5(0.5) & 0.1(0.3) \\
	E-Divisive & & 300.0(0.0) & 300.0(0.1) & 300.0(0.5) & 296.4(16.2) & 287.2(31.4) \\[5pt]
	\hline 
 	\multicolumn{7}{c}{Setting (b)} \\[5pt]
 	& & $T = 180$ & $T = 240$ & $T = 300$ & $T = 360$ & $T = 420$ \\ [3pt]
	\pdp & $\mathcal{D}$ & 11.5(6.2) & 3.7(4.6) & 2.5(4.6) & 2.8(4.3) & 1.2(3.6) \\
	SBS-MVTS & & 177.0(21.1)  & 233.3(38.1)  & 270.1(85.5) & 243.8  (156.1) & 263.5(185.2) \\
	E-Divisive & & 61.0(0.0) & 81.0(0.0)  & 101.0(0.0) & 121.0(0.0) & 141.0(0.0)  \\  [3pt]
	\pdp & $|\widehat{K} - K|$ & 0.0(0.0) & 0.0(0.0) & 0.0(0.0) & 0.0(0.0) & 0.0(0.0) \\
	SBS-MVTS & & 2.0(0.2) & 1.9(0.3) & 1.9(0.4) & 1.6(0.7) & 1.6(0.6) \\
	E-Divisive & & 178.9(0.3) & 238.9(0.3) & 298.9(0.3) & 358.8(0.4) & 418.8(0.4) \\	 [5pt]
	\hline 
 	\multicolumn{7}{c}{Setting (c)} \\[5pt]	
	& & $M = 15$ & $M = 20$ & $M = 25$ & $M = 30$ & $M = 35$ \\ [3pt]
	\pdp & $\mathcal{D}$ & 3.3(5.0) & 3.6(5.5) & 3.2(4.5) & 5.0(12.4) & 6.1(13.2) \\
	SBS-MVTS & & 401.4(112.8) & 378.2(129.9) & 411.3(101.7) & 377.7(134.1) & 375.4(134.5) \\
	E-Divisive & & 151.0(0.0) & 151.0(0.0) & 151.0(0.0) & 151.0(0.0) & 151.0(0.0) \\  [3pt]
	\pdp & $|\widehat{K} - K|$ & 0.0(0.0) & 0.0(0.0) & 0.0(0.0) & 0.0(0.0) & 0.0(0.0) \\
	SBS-MVTS & & 1.8(0.4) & 1.7(0.5) & 1.9(0.3) & 1.8(0.4) & 1.8(0.4) \\
	E-Divisive & & 448.6(0.5) & 449.0(0.3) & 449.0(0.1) & 449.0(0.0) & 449.0(0.0)
\end{tabular}
\end{center}
\caption{Simulation results.  Each cell is in the form of mean(standard error).  For the metrics, $\mathcal{D}$ denotes the Hausdorff distance defined in \eqref{eq-D-defineee} and $|\widehat{K} - K|$ denotes the absolute errors in estimating the numbers of the change points.  \pdp uniformly outperforms the other two methods across a range of parameter values, including $\rho$ (reflecting the jump size), $T$ (the number of samples), and $M$ (the dimension of the time series). \label{tab-sim}}
\end{table}
\begin{figure}[h]
	\includegraphics[width=\columnwidth]{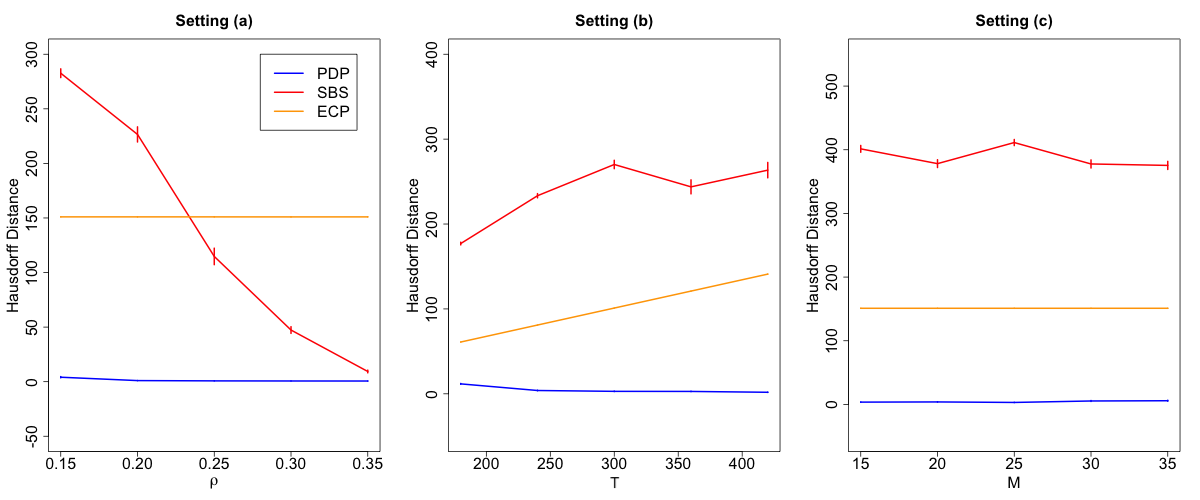}
	\caption{A visualization of the mean Hausdorff distance metric $\mathcal{D}$ \eqref{eq-D-defineee} in \Cref{tab-sim}.  The methods concerned are: \Cref{algorithm:PDP}, \pdp; SBS, SBS-MVTS; ECP, E-Divisive.  In each panel, the y-axis represents the mean Hausdorff distance across 100 repetitions and the x-axis represents the varying parameter in each setting.  \pdp uniformly outperforms the other two methods across a range of parameter values, including $\rho$ (reflecting the jump size), $T$ (the number of samples), and $M$ (the dimension of the time series).}
	\label{fig:simulations}
\end{figure} 
  
We collect the simulation results in \Cref{tab-sim}, each cell containing the mean and standard errors of 100 repetitions.  The Hausdorff distances are visualized in \Cref{fig:simulations} to improve readability.  These three settings have ranged over various situations.  It is clearly that \pdp outperforms both competitors in all settings on both metrics.

\subsection{Real data example}\label{sec-real-data}

We consider the neuron spike train data set previously analyzed in  \cite{watson2016network}.  The three chosen data sets are from \cite{WatsonBO} and each consists of wake-sleep episodes of multi-neuron spike train recording sessions of one laboratory animal.  Each wake-sleep episode includes at least 7 minutes of wake time followed by at least 20 minutes of sleep time.  Note that the wake and sleep periods were recorded so the true change point in each dataset is the end of the wake period.  For each data set, we first compute the Firing Rate (FR) of each neuron using a 5-second discretization time window and then apply \Cref{algorithm:PDP} with $\lambda = 800$ and $\gamma = \log^2(M)/2$, the same as in  \Cref{sec-sim}.  For comparison, we also apply the SBS-MVTS algorithm \citep{cho2015multiple} and E-Divisive procedure \citep{matteson2014nonparametric}.  

These three subsets are on subjects 20140528\_565um, BWRat17\_121912 and BWRat19\_032413.  The numbers of neurons, i.e.~the dimensions of the time series $M$, are 24, 33 and 41, respectively.  The total numbers of 5-second time intervals, i.e.~the total number of time points $T$ considered in \Cref{model:sepp change point}, are 3750, 2995 and 3920, respectively.  The true change points are at point 788, 1184 and 2001, respectively.

The results are summarized in \Cref{table:real data} and are depicted in \Cref{fig:realdata}.  As we can see from the table and the figure, our \pdp algorithm consistently outperforms the other two algorithms in these real data examples. 
  
\begin{table} [ht]  
\begin{center}
	\begin{tabular} {ccccc}  
	Subject & Metric & \pdp & SBS-MVTS & E-Divisive \\[5pt]
	\hline 
	\multirow{2}{*}{20140528\_565um} & $\mathcal D$ & 38 & 382 & 2966 \\
	& $|\widehat K -K|$ & 0 & 0 & 740 \\
	\multirow{2}{*}{BWRat17\_121912} & $\mathcal D$ & 84 & 140 & 1816 \\
	& $|\widehat K -K|$ & 0 & 0 & 595 \\
	\multirow{2}{*}{BWRat19\_032413} & $\mathcal D$ 	& 1 & 99 & 1996 \\
	& $|\widehat K -K|$ & 0 & 0 & 773 
	\end{tabular}
	\caption{The results of three algorithms on multi-neuron spike train data sets.     For the metrics, $\mathcal{D}$ denotes the Hausdorff distance defined in \eqref{eq-D-defineee} and $|\widehat{K} - K|$ denotes the absolute errors in estimating the numbers of the change points.  \pdp uniformly outperforms the other two methods.\label{table:real data}} 
\end{center} 
\end{table}   

\begin{figure}[p]
	\includegraphics[width=\columnwidth]{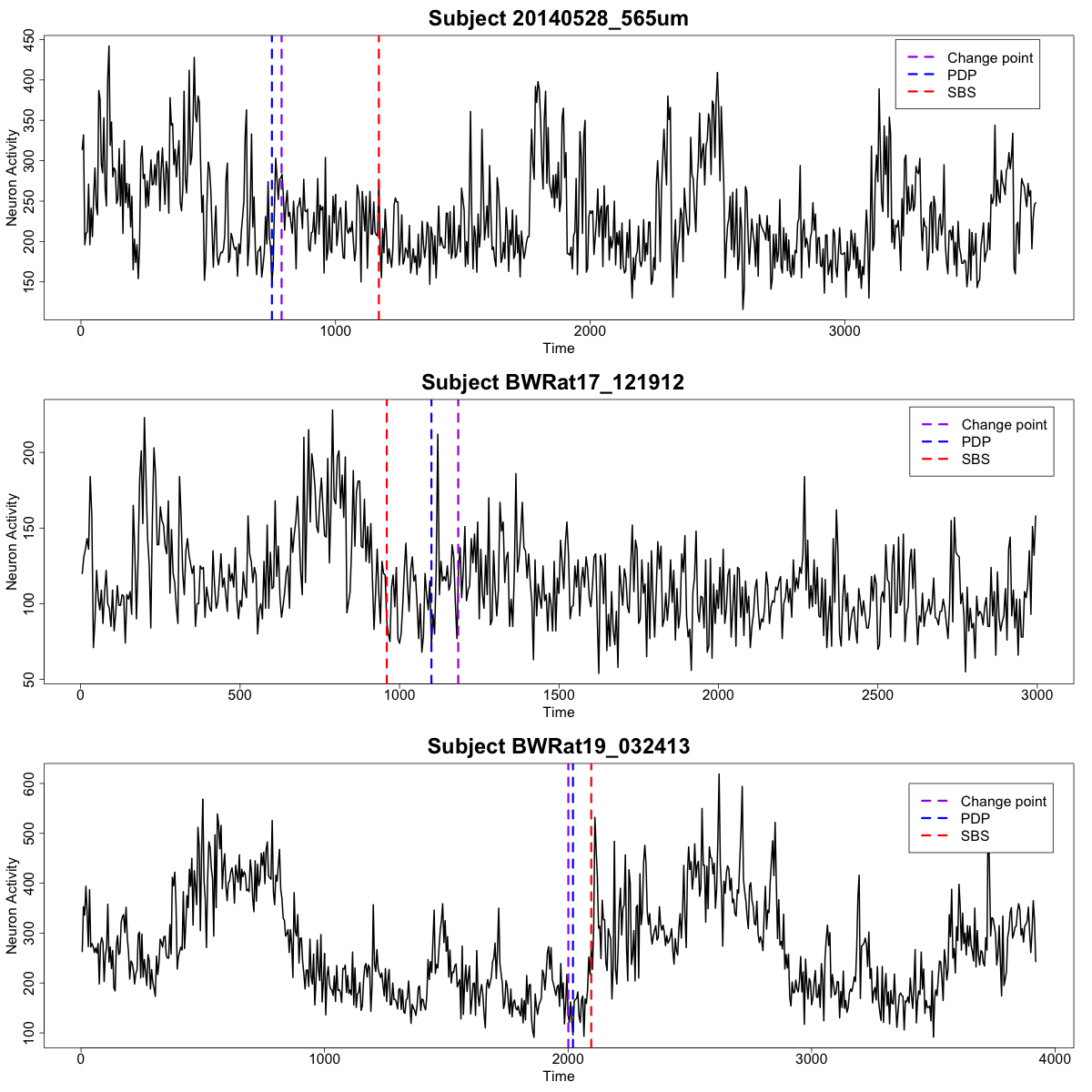}
	\caption{The true change points and the estimators provided by   \pdp and SBS-MVTS in the multi-neuron spike train data sets.  Each panel corresponds to a subject.  The y-axis represents the sum of the FRs across all neurons   and the x-axis represents the ordered time intervals.  The   estimators of the E-Divisive procedure are not included because the corresponding $\widehat K$'s are too large.  \pdp uniformly outperforms the other two methods.  See \Cref{table:real data} for detailed information.  }
	\label{fig:realdata}
\end{figure}   
 
\section{Discussions}

In this paper, we studied piecewise-stationary discrete-time high-dimensional self-exciting Poisson processes, which, or at least the theoretical properties of which were not studied in the literature.  The number of stationary segments in the whole time series is assumed to be an unknown constant.  All the other model parameters are allowed to be functions of the sample size $T$.  We proposed a computationally-efficient and theoretically-guaranteed algorithm.

In the numerical experiments, we fix tuning parameters.  One future research direction is to investigate data-driven methods for tuning parameter selection.  Possible methods include variants of stationary bootstrap \citep{politis1994stationary} or variants of information criteria \citep[e.g.][]{chen2012extended}.

Another future research direction is to extend the techniques we derived in this paper to other popular time series models.  For instance, one key feature of the \sepps we are concerned in this paper is the varying variance structure and heavy tail behaviours.  These share similarities with the GARCH models, which are widely used in finance. 

\bibliographystyle{plainnat}
\bibliography{citations}

\appendix


\section{Proof of \Cref{theorem:SEPP change point}} \label{sec-proof-thm}

In all the appendices, we do not distinguish the notation of every single absolute constant.  For notational simplicity, in the appendices, we drop the subscript of the function $g_t(\cdot)$.

\Cref{theorem:SEPP change point} is an immediate consequence of Proposition~\ref{proposition:dp 1} and \ref{lemma:consistency of number}.

\begin{proposition}\label{proposition:dp 1}
Let $\hatp$ be the defined in \eqref{eq:SEPP DP}.  Under all the assumptions in \Cref{theorem:SEPP change point}, with probability at least $1 - (TM)^{-1}$, the following hold uniformly for any $I = (s, e) \in \hatp$. 
	\begin{itemize}
	\item[{\bf a.}] If $I$ contains only one change point $\eta$, then there exists an absolute constant $C > 0$ such that $$\min \{e-\eta , \, \eta-s  \} \le  C \left ( \frac{ p \lambda^2  d }{ \kappa^2} +  \frac{ p \lambda^2  d^2  }{ \kappa^4}   + \frac{\gamma}{ \kappa^2} \right )    .  $$   
	\item [{\bf b.}] If   $  I  $    contains  exactly two  change points $\eta_{k} $ and $\eta_{k+1}$ , then there exists an absolute constant $C > 0$ such that  $$\max \{e-\eta_{k+1},  \, \eta_{k} -s  \} \le  C \left ( \frac{ p \lambda^2 d }{ \kappa^2} +  \frac{ p \lambda^2  d^2  }{ \kappa^4}   + \frac{\gamma}{ \kappa^2} \right )   .  $$ 
	\item [{\bf c.}] If $|\hatp| > 1$, then let $I$ and $J$ be two consecutive intervals in $\hatp$.  The interval $I \cup J$ contains at least one change point.
	\item [{\bf d.}] The interval $I$ does not contain more than two change points.
	\end{itemize}
\end{proposition} 

\begin{proposition} \label{lemma:consistency of number}
Let $\hatp$ be the defined in \eqref{eq:SEPP DP}.  Assume that $K \le | \hatp| \le 3K$.  Under all the assumptions in \Cref{theorem:SEPP change point}, with probability at least $1 - (TM)^{-1}$, it holds that $| \widehat {\mathcal  P} | = K$.
\end{proposition}

\section{Proof of \Cref{proposition:dp 1}} 
\begin{proof} [Proof of \Cref{proposition:dp 1}]

\Cref{proposition:dp 1} is an immediate consequence of Lemmas~\ref{lemma:one change point}, \ref{lemma:two change point}, \ref{lemma:no change point} and  \ref{lemma:three change point}.  For illustration, we only prove the claim {\bf a}.

Let $I_1 = (s, \eta]$ and $I_2 = (\eta, e]$. Note that
\begin{align}\label{eq:prop 1 case a}  
	 H(\widehat A(I) ,I )  
	\le &  
  H(  \widehat  A  (I_1),I_1 ) +  
   H(  \widehat  A (I_2) , I_2 )   
 + \gamma 
    \le 
    H(     A^*(I_1) ,I_1 ) + H(    A^*(I_2)  ,I_2 )+\gamma
	\end{align}
	where the first inequality follows from 
   the fact that  $\widehat {\p} $ is a minimizer defined in \eqref{eq:SEPP DP} and 
   the second inequality follows  from  the definitions of $\widehat A(I_1) $ and $\widehat A(I_2)$ as in \eqref{eq:definition of A hat}. 
Since there are at most $T^2$ integer intervals in $[1, T]$, a union bound argument leads to that with probability at least $1-(TM)^{-2}$, \Cref{lemma:one change point}  holds uniformly for all integer intervals in $[1, T]$.  The claim \textbf{a.}~therefore holds.
\end{proof}

The following lemma is based on the strong convexity of the log-likelihood function. 
\begin{lemma} \label{eq:strong convexity standard}
For $A \in \mathcal C$, where $\mathcal{C}$ is defined in \eqref{eq:constrain set}, under \Cref{assume:parameters}, for any $m \in [1, M]$ and $t \in [1, T - 1]$, it holds that 
\begin{align*}     
& \exp [v + A_m g\{\mathcal X ( t ) \}]   -        \exp [v +A_m^* (t) g\{\mathcal X ( t ) \}]   - X_m(t+1) \Delta_m(t)  g\{\mathcal X ( t ) \}  \\
& \hspace{6cm}  \ge   c  \left[   \Delta_m(t)  g \{\mathcal X(t) \}  \right]^2 +\epsilon_m(t)     \Delta_m(t)     g\{\mathcal X ( t ) \},
 \end{align*}
 where $  \Delta (t) = A  -A ^* (t)$ and $\epsilon (t) = X(t+1) -  \exp [v +A_m^* (t) g\{\mathcal X ( t ) \}] $.
\end{lemma}

\begin{proof} 

For any $C > 0$ and $a, b\in [-C, C]$, due to the strong concavity of $\exp(\cdot) $, it holds that 
\begin{align*}
	\exp(a) \ge \exp(b) + \exp(a) (b-a) + c(b-a)^2,
\end{align*} 
where $c$ is constant depending on $C$.  Since for any $A\in \mathcal C$ and $m \in [1, M]$,  there exists a constant $C $ depending on $\mathcal C$ and $C_g$ satisfying $|v + A_m g\{\mathcal X(t)\}| \le C$, it holds that 
\begin{align*}
	& \exp [v + A_m  g\{\mathcal X ( t ) \}]    -   \exp [v +A_m^* (t) g\{\mathcal X ( t ) \}]   \\
	& \hspace{2cm} \ge -\exp [v +A_m^* (t) g\{\mathcal X ( t ) \}]  \Delta_m(t)  g\{\mathcal X ( t )\}  + c [ \Delta_m(t)  g\{\mathcal X ( t ) \}]^2,
 \end{align*} 
which leads to the claim. 
\end{proof}

\begin{lemma} \label{lemma:one change point}
For $I = (s, e) \subset (0, T+1)$, assume that $I$ contains only one change point $\eta$.  Denote $I_1 = (s,\eta]   $ and $I_2 = (\eta, e]$.  Assume that $\|A^* (I_1)  -A ^*(I_2)\|_{\mathrm{F}} =\kappa > 0$.  If 
	\begin{align} \label{eq:lemma 1 inequality}
		H\{\widehat A(I) ,I \}  \le & H\{  A^* (I_1),I_1 \} +  H\{ A^* (I_2) , I_2 \}  + \gamma,
	\end{align}
	then with  probability at least $ 1-  4 (TM)^{-4}$, there exists an absolute constant $C > 0$ such that 
	$$\min \{ |I_1| , \, |I_2|  \} \le  C_\epsilon \left( \frac{ \lambda^2 d}{  \kappa^2}   +  \frac{  \lambda^2 d^2 }{  \kappa^4}+\frac{\gamma}{ \kappa^2} \right). $$ 
\end{lemma}

\begin{proof}
Without loss of generality, assume $|I_1| \ge |I_2|$.  \Cref{eq:lemma 1 inequality}  implies that there exits a constant $c > 0$ such that 
	\begin{align}
		& c\sum_{t  \in I }\sum_{m = 1}^M [\Delta_m(t) g\{\mathcal X (t)\}]^2 \le  - \sum_{t  \in I}  \sum_{m = 1}^M \epsilon_m(t) \Delta_m(t)g\{\mathcal X (t)\} \label{eq:basic inequality} \\
		& \hspace{2cm} - \lambda \sqrt {|I | } \| \widehat A\|_1 + \lambda \left\{\sqrt {|I_1|} \|A^*(I_1)\|_1 + \sqrt {|I_2|} \|A^*(I_2)\|_1 \right\} + \gamma, 	\label{eq:basic inequality first term}
	\end{align}
	where $\epsilon_m(t) = X_m(t+1) - \exp[v+ A_m^* (t) g\{\mathcal X (t)\}]$ and $\Delta_m (t)= \widehat A_m - A^*_m(t)$. \\~\\
	
\noindent {\bf Step 1.}  Note that
	\begin{align*}
		& \sum_{t \in I } \sum_{m = 1}^M \epsilon_m(t) \Delta_m(t)  g\{\mathcal X (t)\} = \sum_{t   \in I } \sum_{m, m' = 1}^M \epsilon_m(t) \Delta_{m, m'}(t)  g_{m'}\{\mathcal X (t)\} \\
		= & \sum_{t   \in I }\sum_{(m,m')\in S} \epsilon_m(t) \Delta_{m, m'}(t)  g_{m'}\{\mathcal X (t)\} + \sum_{t \in I }\sum_{(m, m')\in S^c} \epsilon_m(t) \widehat A_{m, m'}  g_{m'}\{\mathcal X (t)\}.
	\end{align*}
	It follows from  \Cref{lemma:finite change points event 1} that with  probability at least $1-(TM)^{-4}$, there exists an absolute constant $C > 0$ such that
	\begin{align*}
		& \sum_{t \in I   }\sum_{(m, m')\in S} \epsilon_m(t) \Delta_{m, m'}(t)  g_{m'}\{\mathcal X (t)\} \\
		\le & \sup_{(m, m')\in S} \left |\frac{\sum_{t   \in I  } \epsilon_m(t)\Delta_{m, m'} (t)  g_{m'} \{\mathcal X (t)\}}{ \sqrt {\sum_{t \in I } \Delta_{m, m'}^2 (t)}} \right| \sum_{(m, m')\in S} \sqrt {\sum_{t \in I}  \Delta_{m, m'}^2 (t)} \\
		\le  & C\log(TM)  \sqrt {d  \sum_{t\in I } \|\Delta(t) \|_{\mathrm{F}}^2}.
	\end{align*}
	It follows from \Cref{lemma:finite change points event 2} that with  probability at least $1-(TM)^{-4}$, there exists an absolute constant $C > 0$ such that
	\begin{align*}
		\sum_{t \in I }\sum_{(m, m')\in S^c} \epsilon_m(t) \widehat A_{m,m'}  g_{m'}\{\mathcal X (t)\} & = \sum_{(m, m')\in S^c}  \widehat A_{m, m'} \sum_{t \in I  }\epsilon_m(t)   g_{m'}\{\mathcal X (t)\} \\
		& \le C\log^{1/2}(TM)\sqrt {|I| }  \|\widehat  A_{S^c} \|_1.
	\end{align*}

\vskip 3mm
\noindent {\bf Step 2.} For \eqref{eq:basic inequality first term}, since both $A^*(I_1)$ and $A^*(I_2)$ are supported on $S$, it follows that
	\begin{align*}
		& -\lambda \sqrt {|I|} \| \widehat A\|_1 + \lambda \left( \sqrt {|I_1|} \|A^*(I_1)\|_1 + \sqrt {|I_2|} \|A^*(I_2)\|_1 \right) \\
		= & -\lambda \sqrt {|I|} (\| \widehat A_S\|_1 +\| \widehat A_{S^c}\|_1) + \lambda \left( \sqrt {|I_1|} \|A^*_S(I_1) \|_1 + \sqrt {|I_2|} \|A^*_S(I_2)\|_1 \right) \\
		\le & - \lambda \sqrt {|I| } \| \widehat A_{S^c}\|_1 + \lambda \sqrt {|I_1|} \| A^*_S(I_1) -\widehat A_S \|_1 + \lambda\sqrt {|I_2|} \|A^*_S(I_2)\|_1 \\
		\le & -\lambda \sqrt {|I| } \| \widehat A_{S^c}\|_1 +\lambda \sqrt {|I_1|d} \|A^*_S(I_1) -\widehat A_S \|_{\mathrm{F}}+ \lambda\sqrt {|I_2|} d \\
		\le & -\lambda \sqrt {|I| } \| \widehat A_{S^c}\|_1 +\lambda \sqrt {d \sum_{t\in I } \|\Delta(t) \|_{\mathrm{F}}^2} + \lambda\sqrt {|I_2|}d. 
	\end{align*}

\vskip 3mm
\noindent {\bf Step 3.} Applying {\bf Steps 1} and {\bf 2} to \eqref{eq:basic inequality} and \eqref{eq:basic inequality first term}, respectively, leads to
	\begin{align}\label{eq:basic inequality 2}
		c \sum_{t \in I}\sum_{m = 1, \ldots, M}  [\Delta_m(t) g\{\mathcal X (t)\}]^2 \le 3 \lambda \sqrt {d\sum_{t\in I } \|\Delta(t) \|_{\mathrm{F}}^2} - 2\lambda/3 \sqrt {|I| } \| \widehat A_{S^c}\|_1 +\lambda \sqrt {|I_2|} d +\gamma,
	\end{align}
	with $\lambda \ge C_\lambda \log(TM)$.  For any $q \in \mathbb Z^+ \cap [0, p-1]$, let $$\mathcal T_q = \{pb+q:\, b\in \mathbb Z^+, \, 1\le p b+q \le T \}.$$
	Then
	\begin{align}
		c \sum_{t \in I \cap \mathcal T_q} \sum_{m = 1, \ldots, M} [\Delta_m(t) g\{\mathcal X (t)\}]^2 = & \sum_{t \in I \cap  \mathcal T_q} \sum_{m = 1, \ldots, M} \Delta_m (t) \mathbb{E}[g\{\mathcal X (t)\}g\{\mathcal X (t)\}^\top | \mathcal X (t-p)]\Delta_m^{\top}(t) \nonumber \\
		- & \sum_{t \in I \cap  \mathcal T_q} \sum_{m = 1, \ldots, M} \Delta_m(t) y (t) \Delta_m^{\top} (t), \label{eq:martingale}
	\end{align}
	where $y(t) = \mathbb{E}[g\{\mathcal X (t)\}g\{\mathcal X (t)\}^\top   | \mathcal X (t-p)] - g\{\mathcal X (t)\}g\{\mathcal X (t)\}^\top   \in \mathbb R^{M\times M}$.  By Assumption~{\bf A4}, \eqref{eq:martingale} implies that
	\begin{align} \label{eq:ar lower bound}
		& c \sum_{t \in I \cap \mathcal T_q} \sum_{m = 1, \ldots, M}[\Delta_m(t) g\{\mathcal X (t)\}]^2 \nonumber \\
		\ge & \sum_{t \in I \cap  \mathcal T_q} \sum_{m = 1, \ldots, M} \|\Delta_m(t)\|_{\mathrm{F}}^2 \left\{\omega - \frac{ \sum_{t \in I \cap  \mathcal T_q} \sum_{m = 1, \ldots, M} \Delta_m (t) y (t) \Delta_m^{\top} (t)}{\sum_{t \in I \cap \mathcal T_q} \sum_{m = 1, \ldots, M} \|\Delta_m(t)\|_{\mathrm{F}}^2}\right\}.
	\end{align}

Now consider
	\[
	Q_1 = \left\{q: \, \omega  \sum_{t \in I \cap  \mathcal T_q} \sum_{m = 1, \ldots, M} \|\Delta_m(t)\|_{\mathrm{F}}^2  \ge 4\sum_{t \in I \cap  \mathcal T_q} \sum_{m = 1, \ldots, M} \Delta_m   (t)y (t) \Delta_m^\top (t)\right\} 
	\]
	and
	\[
	Q_2  = \left\{q: \, \omega  \sum_{t \in I \cap  \mathcal T_q} \sum_{m = 1, \ldots, M} \|\Delta_m(t)\|_{\mathrm{F}}^2  < 4\sum_{t \in I \cap  \mathcal T_q} \sum_{m = 1, \ldots, M} \Delta_m   (t)y (t) \Delta_m^\top (t)\right\}.
	\]

\vskip 3mm
\noindent {\bf Step 4.} For $q\in Q_1$, it follows from \eqref{eq:ar lower bound} and the definition of $Q_1$ that 
	\begin{align}\label{eq:Q1 inequality}
		\sum_{t \in I \cap \mathcal T_q} \sum_{m = 1, \ldots, M} [\Delta_m (t) g\{\mathcal X (t)\}]^2 \ge \omega/2 \sum_{t \in I \cap \mathcal T_q} \sum_{m = 1, \ldots, M}\|\Delta_m(t)\|_{\mathrm{F}}^2.
	\end{align}

Since $\widehat A,  A^*(I_1) \in \mathcal C$, one has for $t \in I_1$,
	\begin{align*}
		& \max_{m = 1, \ldots, M} \sum_{m' = 1, \ldots, M}|\Delta_{m, m'}(t)| \\  
		\le & \max_{m = 1, \ldots, M} \sum_{m' = 1, \ldots, M} | \widehat A _{m, m'}(t)|  + \max_{m = 1, \ldots, M} \sum_{m' = 1, \ldots, M} |A^*_{m, m'}(I_1)| \le 2,
	\end{align*}
	due to Assumption {\bf A2}.  Since $A^*(I_2) \in \mathcal C$, for $t \in I_2$,
	\[
		\max_{m = 1, \ldots, M} \sum_{m' = 1, \ldots, M}|\Delta_{m, m'}(t)| \le 2. 
	\]
	Consequently, 
	\begin{align}
		\max_{m = 1, \ldots, M} \sum_{m' = 1, \ldots, M} \sup_{t\in I } | \Delta_{m, m'}(t)| \leq 2. \label{eq:consequence of stablity}
	  \end{align}

\vskip 3mm	  
\noindent {\bf Step 5.} For $q\in Q_2$, observe that 
	\begin{align} 
		& \sum_{t \in I \cap  \mathcal T_q} \sum_{m = 1, \ldots, M}\Delta_m(t) y (t) \Delta_m^{\top}(t) = \sum_{t \in I \cap  \mathcal T_q} \sum_{m, m', m'' = 1, \ldots, M} \Delta_{m, m'}(t) y_{m', m''} (t)\Delta_{m, m''}(t) \nonumber \\ 
		= & \sum_{t \in I \cap  \mathcal T_q} \sum_{(m, m') \in S} \sum_{m'' = 1, \ldots, M } \Delta_{m, m'}(t) y_{m', m''} (t)\Delta_{m, m''}(t) \label{eq:step 6 first term} \\ 
		+ & \sum_{t \in I \cap  \mathcal T_q} \sum_{(m, m') \in S^c} \sum_{m'' = 1, \ldots, M }  \widehat A_{m, m'} y_{m',m''} (t)\Delta_{m,m''}(t).  \label{eq:step 6 second term}
	\end{align}
	
For \Cref{eq:step 6 first term}, due to \Cref{lemma:deviation bound 2}, it holds with probability at least $1 -(TM)^{-4}$ that
	\begin{align*}
		& \sum_{t \in I \cap  \mathcal T_q} \sum_{(m, m') \in S} \sum_{m'' = 1, \ldots, M} \Delta_{m, m'}(t) y_{m',m''} (t)\Delta_{m,m''}(t) \\
		= & \sum_{(m, m')\in S, \, m'' = 1, \ldots, M} \left\{\frac{\sum_{t \in I \cap \mathcal T_q } \Delta_{m, m'}(t) y_{m', m''} (t)\Delta_{m, m''}(t) } { \sqrt {\sum_{t \in I \cap  \mathcal T_q} \Delta_{m, m'}^2 (t) \Delta_{m, m''}^2 (t)} }\right\}\sqrt {\sum_{t \in I \cap  \mathcal T+q} \Delta_{m,m'}^2 (t) \Delta_{m,m''}^2 (t)} \\
		\le & \max_{m, m', m'' = 1, \ldots, M} \left|\frac{\sum_{t \in I \cap \mathcal T_q } \Delta_{m, m'}(t) y_{m',m''} (t)\Delta_{m,m''}(t)}{ \sqrt {\sum_{t \in I \cap  \mathcal T_q} \Delta_{m, m'}^2 (t) \Delta_{m, m''}^2 (t)} }\right| \sum_{(m, m')\in S, \, m'' = 1, \ldots, M } \sqrt {\sum_{t \in I \cap  \mathcal T_q} \Delta_{m,m'}^2 (t) \Delta_{m,m''}^2 (t)} \\	
		\le & C\log(TM)\sum_{(m, m')\in S, \, m'' = 1, \ldots, M}\sqrt {\sum_{t \in I \cap \mathcal T_q} \Delta_{m, m'}^2 (t) \Delta_{m,m''}^2 (t)} \\
		\le & C\log(TM)  \sum_{(m, m')\in S}\sqrt {\sum_{t \in I \cap  \mathcal T_q} \Delta_{m, m'}^2 (t)} \max_{m = 1, \ldots,  M } \sum_{m'' = 1, \ldots, M }  \sup_{t\in I }  | \Delta_{m, m''}(t)| \\
		\le & 2C\log(TM)  \sqrt {d} \sqrt {\sum_{(m,m')\in S}\sum_{t \in I \cap  \mathcal T_q} \Delta_{m, m'}^2 (t)} \\
		\le & 32 C \omega^{-1} \log^2(TM) d + \frac{\omega}{16}\sum_{(m, m')\in S}\sum_{t \in I \cap  \mathcal T_q} \Delta_{m, m'}^2 (t),
	\end{align*}
	where the second inequality follows from \Cref{lemma:deviation bound 2} and the fourth inequality follows from  \eqref{eq:consequence of stablity}.  For \eqref{eq:step 6 second term}, similarly, we have that with probability at least $1- (TM)^{-4}$,
	\begin{align*}
		& \sum_{t \in I \cap  \mathcal T_q} \sum_{(m,m') \in S^c} \sum_{ m'' = 1, \ldots, M }  \widehat A_{m,m'} y_{m',m''} (t)\Delta_{m,m''}(t) \\
		= & \sum_{ m'' = 1, \ldots, M } \sum_{(m,m') \in S^c} \widehat A_{m,m'}  \sum_{t \in I \cap  \mathcal T_q} y_{m',m''} (t)\Delta_{m,m''}(t) \\
		\le & \sum_{ m'' = 1, \ldots, M } \sum_{(m,m') \in S^c} |\widehat A_{m,m'} | \left|\sum_{t \in I \cap  \mathcal T_q} \frac{y_{m',m''} (t)\Delta_{m,m''}(t)}{\sqrt {\sum _{t\in I \cap \mathcal T+q} \Delta ^2_{m,m''}(t) }}\sqrt {\sum _{t\in I \cap \mathcal T+q} \Delta ^2_{m,m''}(t) } \right| \\
		\le & C \sum_{ m'' = 1, \ldots, M } \sum_{(m,m') \in S^c}| \widehat A_{m,m'}  |  \log(TM) \sqrt  { |I| }  \max_{t \in I}|\Delta_{m,m''}(t)| \le C \|\widehat A_{S^c} \|_1   \sqrt {|I| }  \log(TM)  
	\end{align*}
	where the second inequality follows from  \Cref{lemma:deviation bound 2} and the last inequality follows from   \eqref{eq:consequence of stablity}. 	 Combining the above calculations with \eqref{eq:step 6 first term} and \eqref{eq:step 6 second term}, we have that
	\begin{align}
		& \sum_{t \in I \cap  \mathcal T_q} \sum_{m = 1, \ldots, M } \Delta_m   (t) y (t) \Delta^{\top}_m(t)   \nonumber \\	 
		\le &  32 C \omega^{-1} \log^2(TM) d  + \frac{\omega}{16}\sum_{(m,m')\in S}\sum_{t \in I \cap  \mathcal T_q} \Delta_{m,m'}^2 (t) + C \|\widehat A_{S^c} \|_1   \sqrt {|I| }  \log(TM)  \label{eq:step 6 Q2}
	\end{align}
	Since $q\in Q_2$, it holds that
	\begin{align}
		& \omega  \sum_{q\in Q_2}\sum_{t \in I \cap  \mathcal T_q} \sum_{m = 1, \ldots, M} \|\Delta_m(t)\|_{\mathrm{F}}^2   \nonumber \\
		\le & 4\sum_{q\in Q_2 } \sum_{t \in I \cap  \mathcal T_q} \sum_{m = 1, \ldots, M} \Delta_m   (t)y (t) \Delta_m^\top (t)\nonumber \\
		\le &C p \omega^{-1} \log^2(TM) d  + \frac{\omega}{4} \sum_{q\in Q_2 } \sum_{(m,m')\in S}\sum_{t \in I \cap  \mathcal T_q} \Delta_{m,m'}^2 (t) + \lambda/6	\|\widehat A_{S^c} \|_1   \sqrt {|I|}, \label{eq:Q_2 inequality}
	\end{align}
	where the first inequality follows from  the definition of  $Q_2$, and  the last inequality  follows from  \eqref{eq:step 6 Q2} and  that $\lambda \ge C p \log(TM)$.  \Cref{eq:Q_2 inequality} directly leads to that 
		\begin{align}
			\omega  \sum_{q\in Q_2}\sum_{t \in I \cap  \mathcal T_q} \sum_{m = 1, \ldots, M} \|\Delta_m(t)\|_{\mathrm{F}}^2 \le  2C p \omega^{-1} \log^2(TM)d + \lambda/3 	\|\widehat A_{S^c} \|_1   \sqrt {|I|}. \label{eq:Q_2 inequality 2}
	\end{align}
	Therefore
	\begin{align*}
		& \omega  \sum_{t\in I } \sum_{ m = 1, \ldots, M } \|\Delta_m(t)\|_{\mathrm{F}}^2 \\ 
		= &   \omega  \sum_{q\in Q_1}\sum_{t \in I \cap  \mathcal T_q} \sum_{ m = 1, \ldots, M } \|\Delta_m(t)\|_{\mathrm{F}}^2  + \omega  \sum_{q\in Q_2}\sum_{t \in I \cap  \mathcal T_q} \sum_{ m = 1, \ldots, M } \|\Delta_m(t)\|_{\mathrm{F}}^2  \\
		\le & 2 \sum_{q\in Q_1}	\sum_{t \in I \cap \mathcal T_q} \sum_{ m = 1, \ldots, M } [\Delta_m (t) g\{\mathcal X (t)\}]^2 + 2Cp\omega^{-1} \log^2(TM)d + \lambda/3 	\|\widehat A_{S^c} \|_1 \sqrt {|I|} \\
		\le & \sum_{t \in I } \sum_{ m = 1, \ldots, M }[\Delta_m (t)  g\{\mathcal X (t)\}]^2 +  2C p \omega^{-1} \log^2(TM)d + \lambda/3 \|\widehat A_{S^c} \|_1   \sqrt {|I|}  \\
		\le & 6/c \lambda\sqrt {d \sum_{t\in I } \|\Delta(t) \|_{\mathrm{F}}^2}  - 4\lambda/(3c) \sqrt {|I| } \| \widehat A_{S^c}\|_1 + 2\lambda/c \sqrt {|I_2|}d +\gamma \\
		& \hspace{5cm} +  2C p \omega^{-1} \log^2(TM)d + \lambda/3	\|\widehat A_{S^c} \|_1   \sqrt {|I|},
      \end{align*}
	where the first inequality follows from \eqref{eq:Q1 inequality} and \eqref{eq:Q_2 inequality 2} and the last inequality follows from  \eqref{eq:basic inequality 2}.   The above display together directly yields that
	\begin{align}\label{eq:lasso standard form 3}
		\omega \sum_{t\in I }\| \Delta(t)\|_{\mathrm{F}}^2 + (\lambda/3)	\|\widehat A_{S^c} \|_1   \sqrt {|I|}    \le \frac{pd}{\omega} \lambda ^2 + \lambda \sqrt {|I_2|}d + \gamma.
	\end{align}

\vskip 3mm
\noindent {\bf Step 6.} Observe that
	\begin{align*}
		\sum_{t\in I } \|\Delta(t) \|_{\mathrm{F}}^2 \ge & \inf_{B \in \mathbb{R}^{M \times M}} \sum_{  t\in I } \|B-A^* (t)\|_{\mathrm{F}}^2  = \inf_{B\in \mathbb{R}^{M \times M}}  \left\{ |I_1|\|B-A^*(I_1)\|_{\mathrm{F}}^2  +|I_1|\|B-A^*(I_2)\|_{\mathrm{F}}^2\right\} \\
		= & \frac{|I_1||I_2|}{|I_1| +|I_2|} \kappa^2  \ge \frac{|I_2|}{2} \kappa^2,
	\end{align*}
	where $|I_1|\ge |I_2| $ is used in the last inequality.  Due to \eqref{eq:lasso standard form 3}, it holds that 
	\[
		\frac{1}{2 } |I_2| \omega  \kappa^2 \le \frac{pd}{\omega} \lambda ^2 + \lambda \sqrt {|I_2|}d +\gamma \le \frac{pd  \lambda ^2  }{\omega}  +  \left(    \frac{8 \lambda^2 d^2}{      \omega \kappa^2 }  +  \frac{1}{4  } |I_2| \omega  \kappa^2\right) + \gamma,
	\]
	where the last inequality follows from  H\ou lder's inequality.  We thus have that with probability at least $1 - 4(TM)^{-4}$,
	\[
		|I_2| \le C_\epsilon \left( \frac{ \lambda^2 d}{  \kappa^2}   +  \frac{  \lambda^2 d^2 }{  \kappa^4}+\frac{\gamma}{ \kappa^2} \right).
	\]
\end{proof}

\begin{lemma} \label{lemma:two change point} 
Under all the assumptions in \Cref{theorem:SEPP change point}, let $I = (s, e] \subset [1, T]$ be any interval containing exactly two change points $\eta_{r+1}$ and $\eta_{r+2}$, $I_1 = (s,\eta_{r+1}]$,  $I_2 = (\eta_{r+1}, \eta_{r+2}] $ and $I_3 = ( \eta_{r +2},  e] $.  Let $\kappa_j = \| A^*(I_j) -A^*(I_{j+1}) \|_{\mathrm{F}}  $ for $j = 1, 2$ and 
	 $\kappa = \min \{ \kappa_1, \, \kappa_2\}$.  If 
	\begin{align}\label{eq:two change point inequality} 
		H(\widehat A(I),I  )  \le   \sum_{ j=1}^{3 } H(  A^*(I_{j}) , I_j )     + 2\gamma,
	\end{align}
	then with  probability at least $1- 4 ( TM)^{-4} $ it holds that
	\[
		\max \{ | I_1| ,\, |I_3|\}   \le  C_\epsilon \left( \frac{ \lambda^2 d}{  \kappa^2}   +  \frac{  \lambda^2 d^2 }{  \kappa^4}+\frac{\gamma}{ \kappa^2} \right). 
	\]
\end{lemma}

\begin{proof}
Without loss of generality, assume that $ |I_1 | \ge |I_3| $.  Then \eqref{eq:two change point inequality} implies that
	\begin{align} 
		& c\sum_{t \in I }\sum_{m = 1, \ldots, M } [\Delta_m(t) g\{\mathcal X (t)\}]^2 \le \sum_{t\in I} \sum_{m = 1, \ldots, M } \epsilon_m(t) \Delta_m(t)  g\{\mathcal X (t)\} \nonumber \\ 
		& \hspace{3cm} - \lambda \sqrt {|I| } \| \widehat A(I) \|_1 + \lambda \sum_{j =1}^{3}\left( \sqrt {|I_j|} \|A^* (I_j) \|_1 \right) + 2 \gamma, \label{eq:multiple basic inequality first term}
	\end{align}
	where $\epsilon_m(t) = \exp[v+ A_m(t) g\{\mathcal X (t)\}]-X_m(t+1)$ and $\Delta (t)= \widehat A (I) -A^* (I_j)$ for $t\in I_j$, $j = 1, 2, 3$.

There are two possible scenarios: (1) $\min\{|I_1|, \, |I_2|\} \geq |I_3|$ and (2) $|I_3| > |I_2|$.  We remark that (2) is simpler than (1), so in the sequel we will assume (1).  Let $\widetilde{I}$ and $\widetilde{J}$ be the shorter and longer one between $I_1$ and $I_2$, respectively.

\vskip 3mm	
\noindent {\bf Step 1.}  For the case   Observe that 
	\begin{align*}
		& -\lambda \sqrt {|I| } \| \widehat A(I) \|_1 + \lambda \sum_{j =1}^{3}\left( \sqrt {|I_j|} \|A^*(I_j) \|_1 \right) \\
		\le & -\lambda \sqrt {|I| } \| \widehat A(I) \|_1 + \lambda \sqrt { | \widetilde{J} |} \| A  ^* (\widetilde{J}) \|_1 + 2 \lambda \sqrt { | \widetilde{I} |} d \\
		\le & -\lambda \sqrt { |I  | }  \|\widehat A_{S^c} (I) \|_1 + \lambda \sqrt {|\widetilde{J}|} \| \widehat A_S (I) - A^*(\widetilde{J})\|_1 + 2 \lambda \sqrt { | \widetilde{I} |} d \\
		\le & -\lambda \sqrt { |I  | }  \|\widehat A_{S^c} (I) \|_1 + \lambda \sqrt {d  \sum_{t\in I } \|  \Delta(t) \|_{\mathrm{F}}^2  } + 2 \lambda \sqrt { | \widetilde{I} |} d,
	\end{align*}
	where the first inequality follows from the assumption that $|I_1|\ge |I_3|$ and the fact that $\|A^*(I_j) \| _1 \le d$ for all $j = 1, 2, 3$. 

\vskip 3mm
\noindent {\bf Step 2.} Following from similar arguments as those in the {\bf Step 1} in the proof of \Cref{lemma:one change point}, one has 
	\begin{align*}
		& \sum_{t \in I }  \epsilon_m(t) \Delta_{m, m'}(t)  g_{m'}\{\mathcal X (t)\} \\
		\le & C\log(TM) \sqrt{d\sum_{t\in I } \|\Delta(t)\|_{\mathrm{F}}^2} + C\log^{1/2}(TM) \sqrt{|I|} \|\widehat A_{S^c} (I)\|_1 \\
		\le & \lambda/6 \sqrt {d \sum_{t\in I } \|\Delta(t) \|_{\mathrm{F}}^2} +\lambda/6  \|\widehat A_{S^c} (I )\|_1   \sqrt {|I |},
	\end{align*}
	where $\lambda \ge C_\lambda\log^{3/2}(TM) $ is used in the last inequality. 

\vskip 3mm
\noindent {\bf Step 3.} Following from similar arguments as those in the {\bf Step 5} in the proof of \Cref{lemma:one change point}, one has 
	\begin{align*}
    	\omega \sum_{t\in I } \sum_{m = 1, \ldots, M} \|\Delta_m(t)\|_{\mathrm{F}}^2  \le 2 \sum_{t \in I } \sum_{m = 1, \ldots, M} [\Delta_m (t) g\{\mathcal X (t)\}]^2 + 2C \omega^{-1} p\log^2(TM)d + \lambda/3	 \sqrt {|I | }   \|\widehat A_{S^c} (I )  \|_1   
	\end{align*}

\vskip 3mm  
\noindent {\bf Step 4.}  Combing all the previous steps gives
	\begin{align} \label{eq:two change lasso standard form 3} 
 		\omega \sum_{t\in I } \sum_{m = 1, \ldots, M} \|\Delta_m(t)\|_{\mathrm{F}}^2     + \lambda/3 \sqrt {|I|} \|\widehat A (I )\|_1 \le 4C \omega^{-1} p  \lambda^2 d  + 2 \lambda d\sqrt {|\widetilde{I}|}+ 2 \gamma.
	\end{align}	
	Observe that
	\begin{align*}
		& \sum_{t\in I } \|\Delta(t) \|_{\mathrm{F}}^2 \ge	\sum_{t\in I_1 \cup I_2 } \|\Delta(t) \|_{\mathrm{F}}^2 \ge \inf_{B \in \mathbb{R}^{M \times M}}  \sum_{  t\in I_1 \cup I_2  } \|B-A^* (t)\|_{\mathrm{F}}^2 \\
		= & \inf_{B \in \mathbb{R}^{M \times M}} \{|I_1|\|B-A^*(I_1)\|_{\mathrm{F}}^2  +|I_2|\|B-A^*(I_2)\|_{\mathrm{F}}^2 \} = \frac{|I_1||I_2|}{|I_1| +|I_2|} \kappa^2 \ge  \frac{\min\{|I_1| ,\, |I_2| \}}{2}   \kappa^2. 
	\end{align*}
	If $|I_2| \leq |I_1|$, then it follows that
	\begin{align*}
		\frac{|I_2| \kappa^2 \omega}{2} \leq 4C \omega^{-1} p  \lambda^2 d  + 2 \lambda d\sqrt {|I_2|}+ 2 \gamma \leq 4C \omega^{-1} p  \lambda^2 d  + \frac{\Delta \kappa^2 \omega }{4} + \frac{8 \lambda^2 d^2}{\kappa^2 \omega} + 2 \gamma.
	\end{align*}
	This aleads to that 
	\[
		|I_2| \lesssim \frac{\lambda^2d}{\kappa^2} + \frac{\lambda^2 d^2}{\kappa^4} + \frac{\gamma}{\kappa^2},
	\]
	which contradicts with Assumption~{\bf A3}.  Then there exits an absolute constant such that 
	\[
		| I_1|    \le  C_\epsilon \left( \frac{  \lambda^2 d}{ \kappa^2} +   \frac{ p \lambda^2  d^2}{ \kappa^4}   + \frac{\gamma}{ \kappa^2} \right).
	\]
	Since by assumption, $|I_3 | \le |I_1|$, the desired results follows.   
\end{proof}

\begin{lemma}  \label{lemma:no change point}
Under all the assumptions in \Cref{theorem:SEPP change point}, let $I = (s, e] \subset [1, T]$ be any interval which contains no change point.  Let $I_1$ and $I_2$ be two intervals such that $I_1 \sqcup I_2 = I$.   Then with probability at least $1 - (TM)^{-4}$,  
	\begin{align*} 
		H(\widehat A (I_1) ,I_1 ) + H(\widehat A (I_2) ,I_2 ) + \gamma \geq  H(  A^*(I ) ,I ).
	\end{align*}
\end{lemma}

\begin{proof}
We prove by contradiction, assuming that 
	\[
		H(\widehat A (I_1) ,I_1 ) + H(\widehat A (I_2) ,I_2 ) + \gamma <  H(  A^*(I ) ,I ).
	\]
	Denote 
	\[
		\Delta (t) = \begin{cases}
			\widehat A (I_1) -A^*(I), & t\in I_1; \\
			\widehat A (I_2) -A^*(I), & t\in I_2.
		\end{cases}
	\]
	Standard calculations give
	\begin{align}
		& c\sum_{t \in I }\sum_{m = 1, \ldots, M} [\Delta_m(t) g\{\mathcal X (t)\}]^2 +\gamma \le \sum_{t \in I }  \sum_{m = 1, \ldots, M} \epsilon_m(t) \Delta_m(t) g\{\mathcal X (t)\} \label{eq:no change basic inequality} \\
		& \hspace{3cm} +  \lambda \sqrt {|I| } \|  A^*(I) \|_1 - \lambda \left\{\sqrt {|I_1|} \|\widehat A (I_1)\|_1 + \sqrt {|I_2|} \|\widehat A (I_2)\|_1 \right\}, \label{eq:no change basic inequality first term}
	\end{align}
	where $\epsilon_m(t) = \exp[v+ A_m(t) g\{\mathcal X (t)\}]-X_m(t+1) $.

\vskip 3mm
\noindent {\bf Step 1.} For \eqref{eq:no change basic inequality first term}, we have that 
	\begin{align*}
		& \lambda \sqrt {|I|} \|  A^*(I)\|_1 - \lambda \left\{\sqrt {|I_1|} \|\widehat A (I_1)\|_1 + \sqrt {|I_2|} \|\widehat A (I_2)\|_1 \right\} \\
		\le & \lambda ( \sqrt {|I_1| } +  \sqrt {|I_2| } )\|  A^*(I)\|_1 -\lambda \left\{\sqrt {|I_1|} \| \widehat A_S (I_1)\|_1 + \sqrt {|I_2|} \| \widehat A_S (I_2)\|_1 \right\} \\
		& \hspace{3cm} - \lambda \left\{\sqrt {|I_1|} \| \widehat A_{S^c} (I_1) \|_1 + \sqrt {|I_2|} \|\widehat A_{S^c} (I_2)\|_1 \right\} \\
		\le & \lambda \left\{\sqrt {|I_1|} \|  \widehat A_S  (I_1)  -A^*(I) \|_1 + \sqrt {|I_2|} \| \widehat A_{S} (I_2)    -A^*(I) \|_1 \right\} \\
		& \hspace{3cm} - \lambda \left\{\sqrt {|I_1|} \| \widehat A _{S^c} (I_1)\|_1 + \sqrt {|I_2|} \|\widehat A_{S^c} (I_2) \|_1 \right\} \\
		\le & \lambda   \sqrt { 2  |I_1|\|  \widehat A_S (I_1)  - A^*(I) \|_1^2 + 2 |I_2| \| \widehat A_S (I_2) -A^*(I) \|_1^2 } \\
		& \hspace{3cm} -  \lambda \left\{\sqrt {|I_1|} \| \widehat A_{S^c}  (I_1)\|_1 + \sqrt {|I_2|} \| \widehat A_{S^c}  (I_2) \|_1 \right\} \\
		\le & \lambda   \sqrt { 2  |I_1|d \|  \widehat A_S (I_1)  - A^*(I) \|_{\mathrm{F}}^2 + 2 |I_2| d \| \widehat A_S (I_2) -A^*(I) \|_{\mathrm{F}}^2 } \\
		& \hspace{3cm} -  \lambda \left\{\sqrt {|I_1|} \| \widehat A_{S^c}  (I_1)\|_1 + \sqrt {|I_2|} \| \widehat A_{S^c}  (I_2) \|_1 \right\} \\
		= & \sqrt { 2\lambda d  \sum_{t\in I } \|  \Delta (t)  \|_{\mathrm{F}}^2 } - \lambda \left\{\sqrt {|I_1|} \| \widehat A_{S^c} (I_1)\|_1 + \sqrt {|I_2|} \| \widehat A_{S^c} (I_2) \|_1 \right\}.
	\end{align*}

\vskip 3mm
\noindent {\bf Step 2.} Using similar calculations as {\bf Step 1} in the proof of \Cref{lemma:one change point}, one has that with probability at least $1 - (TM)^4$,
	\begin{align*}
		& \sum_{t \in I   }  \epsilon_m(t) \Delta_{m,m'}(t)  g_{m'}\{\mathcal X (t)\} \\
		\le & C\log(TM)  \sqrt {d \sum_{t\in I } \|\Delta(t) \|_{\mathrm{F}}^2} + C\log^{1/2}(TM)  \left\{\|\widehat A_{S^c} (I_1)  \|_1   \sqrt {|I_1|} +\|\widehat A_{S^c} (I_2)  \|_1   \sqrt {|I_2|}   \right\} \\
		\le & \lambda/6 \sqrt {d \sum_{t\in I } \|\Delta(t) \|_{\mathrm{F}}^2} +\lambda/6 \left\{\|\widehat A_{S^c} (I_1)  \|_1   \sqrt {|I_1|} +\|\widehat A_{S^c} (I_2)  \|_1   \sqrt {|I_2|}   \right\},
	\end{align*}
	where $\lambda\ge C_\lambda\log(TM) $ is used in the last inequality. 

\vskip 3mm
\noindent {\bf Step 3.} Using similar calculations as  {\bf Step 5}  in the proof of \Cref{lemma:one change point}, one has that with probability at least $1 - (TM)^4$,
	\begin{align*}
    	& \omega \sum_{t\in I } \sum_{m = 1, \ldots, M } \|\Delta_m(t)\|_{\mathrm{F}}^2  \le  \omega/c \sum_{t \in I } \sum_{m = 1, \ldots, M }[\Delta_m (t) g\{\mathcal X (t)\}]^2 \\ 
		& \hspace{2cm} +  2C \omega^{-1} p  \log^2(TM)d 	+ \lambda/3 \left\{ \|\widehat A_{S^c} (I_1)  \|_1   \sqrt {|I_1|} +\|\widehat A_{S^c} (I_2)  \|_1   \sqrt {|I_2|}   \right\}.
      \end{align*}

\vskip 3mm
\noindent {\bf Step 4.} Combing all the previous steps gives
	\begin{align*} 
		\omega  \sum_{t\in I } \sum_{m = 1, \ldots, M } \|\Delta_m(t)\|_{\mathrm{F}}^2    + \gamma + \lambda/3 \left\{\sqrt {|I_1|} \|\widehat A (I_1)\|_1 + \sqrt {|I_2|} \|\widehat A (I_2)\|_1 \right\} \le 4C \omega^{-1} p  \log^2(TM)d,
	\end{align*}
	which directly implies that $\gamma \le 2C_\epsilon  p\lambda^2 d$.  This leads to a contradiction with \eqref{eq-tuning-order} and completes the proof.
\end{proof}

\begin{lemma}\label{lemma:three change point}
Under all the assumptions in \Cref{theorem:SEPP change point}, let $I =(s, e] \subset [1, T]$ satisfying $I \cap  \{ \eta_k\}_{k=1}^K = \{\eta_{r+1},\ldots, \eta_{r+J} \} $, with $J\ge 3$.  Denote  $I_1 = (s,\eta_{r+1}]   $, $I_j = (\eta_{r+j-1}, \eta_{r+j}] $ for all $j = 2, \ldots, J$ and $I_{J+1} = (\eta_{r+J}, e] $.  Denote $\kappa_j = \| A^*(I_j) -A^*(I_{j+1}) \|_{\mathrm{F}}  $ for $j = 1, \ldots, J$ and $\kappa = \min_{j = 1, \ldots, J}  \kappa_j$.  It holds that  probability at least $1- (TM)^{-4}  $,
	\begin{align} \label{eq:minimizer property three chnge points}
		H(\widehat A(I),I  )  > \sum_{ j=1}^{J+1 } H(  A^*(I_{j}) , I_j )    + J \gamma.
	\end{align}
\end{lemma}

\begin{proof}
We prove by contradiction, assuming that 
	\[
		H(\widehat A(I),I  )  \leq \sum_{ j=1}^{J+1 } H(  A^*(I_{j}) , I_j )    + J \gamma.
	\]
	Note that by assumption {\bf A3} $ \min \{ | I_2|, |I_3| \} \ge cT. $
Without loss of generality, assume that $ |I_2 | \ge |I_3| $.
	\Cref{eq:minimizer property three chnge points} implies that
	\begin{align} 
	\sum_{t \in I }\sum_{1\le m \le M } (\Delta_m(t) g(\mathcal X (t)))^2 
	\le 
	&
	\sum_{t\in 	I } 
	\sum_{1\le m \le M }
	\epsilon_m(t) \Delta_m(t)^\top  g(\mathcal X (t)) 
	\\- 
	&
	\lambda \sqrt {|I| } \| \widehat A(I) \|_1 + 
	\lambda \sum_{j =1}^{J+1}\left( \sqrt {|I_j|} \|A^* (I_j) \|_1 \right)
	\label{eq:multiple basic inequality first term} + \gamma
	\end{align}
	where $\epsilon_m(t) = \exp(v+ A_m(t) g(\mathcal X (t)))-X_m(t+1) $
	and $\Delta (t)= \widehat A (I) -A^* (I_j)$ for $t\in I_j \subset I$.
	\\
	\\
	{\bf Step 1.}
	Observe that by assumption {\bf A3} there exists $c>0 $ such that $|I_3|\ge c T \ge c|I_j| $ for any $1\le j\le J$. Therefore 
	\begin{align*}
	&	-\lambda \sqrt {|I| } \| \widehat A(I) \|_1 + 
	\lambda \sum_{j =1}^{J+1}\left( \sqrt {|I_j|} \|A^*(I_j) \|_1 \right)
	\\
\le &
-\lambda \sqrt {|I| } \| \widehat A(I) \|_1 + \lambda \sqrt {|I_{2}| } \| A  ^* (I_2) \|_1 
+ 2 \sum_{j\not = 2 }\lambda \sqrt { | I_{j} |} |S|
\\
\le 
& -\lambda \sqrt { |I  | }  \|\widehat A_{S^c} (I) \|_1 + \lambda \sqrt {|I_2|} \| \widehat A_S (I) - A^*(I_2)\|_1
+ c^{-1/2} \lambda \sqrt { | I_{3} |} |S|
\\
\le 
& -\lambda \sqrt { |I  | }  \|\widehat A_{S^c} (I) \|_1 + \lambda \sqrt { |S|  \sum_{t\in I } \|  \Delta(t) \|_{\mathrm{F}}^2  } 
+ c^{-1/2} \lambda \sqrt { | I_{3} |} |S|
	\end{align*}
	where the first inequality follows from the fact that  	$\|A^*(I_j) \| _1 \le |S| $ for all $1\le  j\le J $ and the second inequality  follows from the observation that $|I_3| \ge c|I_j| $.
	\\
	\\
{\bf Step 2.} Using similar calculations as {\bf Step 1} in the proof of \Cref{lemma:one change point},
one has 
\begin{align*}
	& \sum_{t \in I   }  \epsilon_m(t) \Delta_{m,m'}(t)  g_{m'}(\mathcal X (t)) 
\\
	\le  &
	C\log^{3/2}(TM)  \sqrt {|S|     \sum_{t\in I } \|\Delta(t) \|_{\mathrm{F}}^2}
	+ 
	C\log^{3/2}(TM)  \sqrt{|I| } \|\widehat A_{S^c} (I)  \|_1  
\\
\le & 
(\lambda/6) \sqrt {|S|     \sum_{t\in I } \|\Delta(t) \|_{\mathrm{F}}^2}+
(\lambda/6)   \|\widehat A_{S^c} (I )  \|_1   \sqrt {|I |}  
	\end{align*}
where 
$\lambda\ge C_\lambda\log^{3/2}(TM) $ is used in the last inequality. 
\\
\\	
{\bf Step 3.} Using similar calculations as   {\bf Step 6}  in the proof of \Cref{lemma:one change point}, one has 
   \begin{align*}
    \omega   \sum_{t\in I } \sum_{1\le m \le M } \|\Delta_m(t)\|_{\mathrm{F}}^2  
\le 
	  	\sum_{t \in I } 
	\sum_{1\le m \le M }(\Delta_m (t)^\top   g(\mathcal X (t)))^2   
	+  2C \omega^{-1} p  \log^2(TM) |S|  
	+ (\lambda/3)	 \sqrt {|I | }   \|\widehat A_{S^c} (I )  \|_1   
      \end{align*}
\
\\      
  {\bf Step 4.}    Combing all the previous steps gives
\begin{align}\label{eq:three change lasso standard form 3} 
 \omega   \sum_{t\in I } \sum_{1\le m \le M } \|\Delta_m(t)\|_{\mathrm{F}}^2   +
 (\lambda/3)  \sqrt {|I|} \|\widehat A (I )\|_1 
 \le 4C \omega^{-1} p  \lambda^2  |S|  + 2 \lambda |S| \sqrt {|I_3|} +J \gamma. 
 \end{align}	
 Observe that
	\begin{align*}
	\sum_{t\in I } \|\Delta(t) \|_{\mathrm{F}}^2 \ge	\sum_{t\in I_2 \cup I_3 } \|\Delta(t) \|_{\mathrm{F}}^2 \ge
	&
	\inf_{B} 
\sum_{  t\in I_2 \cup I_3  } \|B-A^* (t)\|_{\mathrm{F}}^2 
	= 	\inf_{B}  
		|I_2|\|B-A^*(I_2)\|_{\mathrm{F}}^2  +|I_3|\|B-A^*(I_3)\|_{\mathrm{F}}^2
		\\
		=& \frac{|I_2||I_3|}{|I_2| +|I_3|} \kappa^2  \ge  
          \frac{\min \{|I_2| , |I_3|  \} }{2 }		
	  \kappa^2 \ge  \frac{c|I_3|}{2 } \kappa^2  ,
	\end{align*}
	where $ |I_2| \ge |I_3|\ge c|I_2| $ is used in the last inequality. 
	So \eqref{eq:three change lasso standard form 3} implies that
	$$ \frac{c}{2} |I_3| \omega  \kappa^2 \le \frac{Cp|S|}{\omega} \lambda ^2 + \lambda \sqrt {|I_3|}|S| + J \gamma 
	\le \frac{Cp|S|}{\omega} \lambda ^2  + \left(  \frac{8 \lambda^2 |S|^2 }{c\omega \kappa^2 }   
+ 	
	 \frac{c}{4} |I_3| \omega  \kappa^2  \right)    + J \gamma .$$ 
	Note that $J \le K \le c^{-1}$ and therefore  this implies
	$$|I_3| \le C _\omega  \left( \frac{ p \lambda^2  |S| }{ \kappa^2}   +   \frac{  p \lambda^2  |S| ^2 }{  \kappa^4}+\frac{\gamma}{  \kappa^2} \right) . $$
	Since by assumption, $|I_3 | \le |I_2|$.  This leads to the contradiction with Assumption~{\bf A3} and completes the proof.
\end{proof}

\section{Proof of \Cref{lemma:consistency of number}}

\begin{proof}[Proof of \Cref{lemma:consistency of number}] 

For a collection of generic strictly increasing time points $\{\eta_j'\}_{j=0}^{J+1}$, where $\eta_0' = 1$ and $\eta_{J+1}' = T + 1$, denote $I_j = [\eta_{j-1}', \eta_j')$ and 
	\[
		\mathcal L (\{\eta_j'\}_{j=0}^{J+1}) = \sum_{j=1}^{ J+1}  H( \widehat A(I_{j} ) ,I_{j} ).
	\]
	In addition assume that $ \{\eta_k \}_{k=1}^K \subset\{\eta_j'\}_{j=1}^{J}$ so that $A^*(t)$ is unchanged in each of the interval $I_j$.  Let
	$$ \mathcal L^* (\{\eta_j'\}_{j=0}^{J+1} )  = \sum_{j=1}^{ J+1}   H(  A^*(I_{j} ) ,I_{j}   ). $$
 
	Let $\{ \widehat \eta_{k}\}_{k=1}^{\widehat K}$ denote the change points induced by $\widehat {\mathcal P}$.  If one can show that 
	\begin{align} 
		& \mathcal L ^*   (\eta_1,\ldots,\eta_K ) +K\gamma \nonumber \\
		\ge  &\mathcal L (\eta_1,\ldots,\eta_K)   + K\gamma \label{eq:K consistency step 1} \\ 
		\ge & \mathcal L   (\widehat \eta_{1},\ldots, \widehat \eta_{\widehat K } ) +\widehat K \gamma  \label{eq:K consistency step 2} \\ 
		\ge  &   \mathcal L^* ( \textbf {Sort } (  \widehat \eta_{1},\ldots, \widehat \eta_{\widehat K } , \eta_1,\ldots,\eta_K ) ) + \widehat K \gamma    - C dK  \lambda^2    - C   K \lambda ^2 d \sqrt {     \frac{ \lambda^2d }{\kappa^2 }      +     \frac{ \lambda^2d^2 }{\kappa^4  }   + \frac{\gamma}{\kappa^2} } \label{eq:K consistency step 3}
	\end{align}
	and that 
	\begin{align}\label{eq:K consistency step 4}
		\mathcal L ^*  (\eta_1,\ldots,\eta_K )   \le  \mathcal L^* ( \textbf {Sort } (  \widehat \eta_{1},\ldots, \widehat \eta_{\widehat K } , \eta_1,\ldots,\eta_K ) ),   
	\end{align}
	then it must hold that $| \hatp | = K$.  To see this, if $\widehat K \ge K+1 $, then  
	\begin{align} \label{eq:Khat=K}  
		\gamma \le  ( \widehat K - K)\gamma  \le  Cd K \lambda^2+    C   K \lambda ^2 d  \sqrt {     \frac{ \lambda^2d }{\kappa^2 }      +     \frac{ \lambda^2d^2 }{\kappa^4  }   + \frac{\gamma}{\kappa^2} },
	\end{align} 
 	which is a contradiction to \eqref{eq-tuning-order} if $C_\gamma$ is sufficiently large. 

Since $ \widehat A(I) = \argmin_{A \in \mathcal C } H(A,I)$, we have that 
	\[
		H (\widehat  A(I), I ) \le  H (A^*(I), I ). 
	\]
	Then \eqref{eq:K consistency step 1}  is a direct consequence  of the above display. 
	
Moreover, \eqref{eq:K consistency step 2} is a direct consequence of \eqref{eq:SEPP DP}.

To show \eqref{eq:K consistency step 3}, consider any $   I =(s,e] \in \hatp $.  By \Cref{proposition:dp 1}, with probability at least $1 - C(TM)^4$, $I$ contains at most two change points. Therefore the three cases in \Cref{lemma:prop2 three cases} directly lead to 
	\[
		\mathcal L   (\widehat \eta_{1},\ldots, \widehat \eta_{\widehat K } )  \ge     \mathcal L^* ( \textbf {Sort } (  \widehat \eta_{1},\ldots, \widehat \eta_{\widehat K } , \eta_1,\ldots,\eta_K ) ) -C K  d\lambda^2  - C   K \lambda ^2 d\sqrt {     \frac{ \lambda^2d }{\kappa^2 }      +     \frac{ \lambda^2d^2 }{\kappa^4  }   + \frac{\gamma}{\kappa^2} }.
	\]

Lastly, suppose $I$ is any generic interval in $[1, T]$ containing no change points and that $ I_1 \sqcup I_2 = I$.  For any $a, b >0$,  we have that $\sqrt {a+b} \le \sqrt {a} + \sqrt {b}$.  This inequality directly implies that 
	\begin{align} \label{eq:direct inequality}
		H(A^*(I),I  )\le H(  A^* ( I_1 ),   I_1 ) + H(  A^* ( I_2 ),   I_2 ).
	\end{align}
	Note that \eqref{eq:K consistency step 4} is a straight forward consequence of \Cref{eq:direct inequality}.  This completes the proof.
\end{proof}

\begin{lemma}[Standard GLM inequality]\label{eq:standard Lasso}
Under all the assumptions in \Cref{theorem:SEPP change point}, suppose $I$ is any generic interval such that $|I| \ge Cp|S| \log^2(TM)$ for sufficiently large $C$ and that $I$ contains no change points.  Let
	\[
		\widehat A(I) = \arg\min_{A\in \mathcal C} H(A,I),
	\]
	where $ H(A,I)$ is defined in \eqref{eq:penalized likelihood}.  Then 
	\[
		\| \widehat A(I) -A^*(I)\|_{\mathrm{F}}^2 \le \frac{Cd\log^2(TM) }{|I|} \quad  \text{and} \quad \| \widehat A(I) -A^*(I)\|_1 \le \frac{Cd \log(TM) }{\sqrt {|I|}}. 
	\]
\end{lemma}

\Cref{eq:standard Lasso} is an immediate consequence of \Cref{eq:strong convexity standard} and results in \Cref{sec-prob-events}, based on standard Lasso arguments \citep[e.g.][]{wang2019localizing}.
%
%

\begin{corollary}\label{coro:SEPP RSS}
Under all the assumptions in \Cref{theorem:SEPP change point} and notation in \Cref{eq:standard Lasso}, it holds that with probability at least $1 - C(TM)^{-4}$,
	\[
		H(\widehat A(I) , I ) \ge  H(A^*(I), I)  + C \log^{2}(TM) d. 
	\]
\end{corollary}

\begin{proof} Denote that $\Delta = \widehat A(I) -A^*(I)$ and that $\epsilon_m  = \exp[v+  A_m ^* (I)g\{\mathcal{X}(t)\}] - X_{m}(t+1)$ .  Observe that 
	\begin{align*}
		& \sum_{t\in I } \sum_{m  = 1, \ldots, M } \exp[v+\widehat A_m (I) g\{\mathcal X(t)\}] - \exp[v+  A_m ^* (I) g\{\mathcal X(t)\}] -X_m(t+1) \Delta_m g\{\mathcal X(t)\} \\
		\ge & \sum_{t\in I }  \sum_{m  = 1, \ldots, M } \epsilon_m \Delta_m g\{\mathcal X(t)\} +c_ \sigma\left[ \{\widehat A_m (I)-  A_m ^*(I) \}g\{\mathcal X(t) \} \right]^2 \\
		\ge & \sum_{t\in I }  \sum_{m  = 1, \ldots, M } \epsilon_m \Delta_m g\{\mathcal X(t)\} \ge -C \sqrt {|I|}\log (TM) \sum_{m  = 1, \ldots, M } \| \Delta_m\|_1  \ge -C \log^2(TM)d, 
	\end{align*}
	where the third inequality follows from \Cref{lemma:finite change points event 2}  and the last inequality follows from  \Cref{eq:standard Lasso}.
\end{proof}

\begin{lemma} \label{lemma:prop2 three cases} 

Under all the conditions in \Cref{theorem:SEPP change point}, suppose $I =(s,e]$ is any interval in $[1, T]$ containing at most two change points.  

\noindent {\bf Case 1.} If $I $ contains no change points, then with probability at least $1- (TM)^{-2}$,
	\[
		H(  A^*(   I ),  I )   \le  H( \widehat A (  I ),  I )  +  Cd \lambda^2 +   C \lambda ^2d  \sqrt {     \frac{ \lambda^2d}{\kappa^2 }      +     \frac{ \lambda^2d ^2 }{\kappa^4  }   + \frac{\gamma}{\kappa^2} }.
	\]

\noindent {\bf Case 2.} If $ I $ contains exactly one change point $\eta_k$, letting $J_1= (s,\eta_k]$ and  $J_2=(\eta_k, e]$, then with probability at least $1- (TM)^{-2} $,
	\[
		H(  A^*(J_1),J_1) +   H(  A^*(J_2),J_2)  \le   H( \widehat A (   I ),  I ) +C d \lambda^2 +   C \lambda ^2 d  \sqrt {     \frac{ \lambda^2d}{\kappa^2 }      +     \frac{ \lambda^2d ^2 }{\kappa^4  }   + \frac{\gamma}{\kappa^2} }    .
	\]
	
\noindent {\bf Case 3.} If $  I $ contains exactly exactly two change points $\eta_k$ and $\eta_{k+1}$, letting $J_1 = (s,\eta_k] $, $J_2 = (\eta_k, \eta_{k+1}]$ and $J_3 = (\eta_{k+1},e]$, then with probability at least $1- (TM)^{-2}$, 
	\[
		\sum_{j=1}^3   H(A^*(J_j),J_j  )\le H( \widehat A ( I ),   I ) +C d \lambda^2   +C \lambda ^2d \sqrt {     \frac{ \lambda^2d}{\kappa^2 }      +     \frac{ \lambda^2d^2 }{\kappa^4  }   + \frac{\gamma}{\kappa^2} } . 
	\]
\end{lemma}

\begin{proof}
We only prove {\bf Case 3}, as the other two cases  are easier and similar.   Since $|J_2| \ge \Delta \ge Cd^2\log^2(TM)$, Observe that by definition of $  H$,
	\begin{align} \label{eq:middle segment}
		H(A^*(J_2),J_2 ) \le H( \widehat A (J_2),J_2 )  + Cd\lambda^2 \le    H( \widehat A ( I ),J_2 ) + Cd\lambda^2,
	\end{align}
	where the first inequality follows from \Cref{coro:SEPP RSS} and the second inequality follows from the definition of $\widehat A(J_2)$.

In addition, observe that with probability at least $1 - C(TM)^{-4}$,
	\begin{align} 
		& H (A^*(J_1), J_1) - \sum_{t\in J_1} \sum_{m = 1, \ldots, M }  \exp[v+  \widehat A_m    (I) g\{\mathcal X(t)\}] -\lambda \sqrt {|J_1| }\| A^*(J_1) \|_1  \nonumber \\
		= & \sum_{t\in J_1} \sum_{m = 1, \ldots, M } \exp[v+  A_m ^*    (J_1) g\{\mathcal X(t)\}] - \exp[v+  \widehat A_m    (I) g\{\mathcal X(t)\}]+ \lambda \sqrt {|J_1| }\| A^*(J_1) \|_1 \nonumber \\
		\le &\sum_{t\in J_1} \sum_{m = 1, \ldots, M } \{ \widehat A_m (I)- A_m ^*(  J_1 )\} g\{\mathcal X(t)\} \epsilon_m(t) - c   \left[\{\widehat A_m (I)-  A_m ^*(J_1) \}g\{\mathcal X(t)\} \right]^2  + \lambda d \sqrt {|J_1|} \nonumber \\
		\le &  \|  \widehat A  (I)- A  ^*(J_1)\|_1 \max_{m, m' = 1, \ldots,  M}\left | \sum_{t\in J_1 }g_{m'}\{\mathcal X(t) \}\epsilon_m(t) \right|  + \lambda d \sqrt {|J_1|}  \nonumber  \\
		\le &  C \lambda d  \sqrt {|J_1| }+ \lambda d\sqrt {|J_1|} \nonumber \\ 
  		\le &    C   \lambda ^2 d \sqrt {     \frac{ \lambda^2d }{\kappa^2 }      +     \frac{ \lambda^2d ^2 }{\kappa^4  }   + \frac{\gamma}{\kappa^2} }, \label{eq:first segment}
	\end{align}
	where the first inequality follows from the same argument in  \Cref{eq:strong convexity standard},  the third inequality follows from \Cref{lemma:finite change points event 2} and the last inequality follows from  \Cref{proposition:dp 1} {\bf b}. 

In addition, it holds that
	\begin{align} \label{eq:last segment}
		H (A^*(J_3), J_3) - \sum_{t\in J_3} \sum_{ m = 1}^M \exp[v+  \widehat A_m    (I) g\{\mathcal X(t)\}] \le   C   \lambda ^2 d \sqrt {     \frac{ \lambda^2d}{\kappa^2 }      +     \frac{ \lambda^2d^2 }{\kappa^4  }   + \frac{\gamma}{\kappa^2} }.
	\end{align}
Therefore
	\begin{align*} 
		& \sum_{j=1}^3   H(A^*(J_j),J_j  ) \\
		\le & H( \widehat A (   I ), J_2 )  +  \sum_{t\in J_1 \cup J_3} \sum_{m = 1}^M  \exp[v+  \widehat A_m    (I) g\{\mathcal X(t)\}] + C \lambda ^2 d \sqrt {     \frac{ \lambda^2d }{\kappa^2 }      +     \frac{ \lambda^2d^2 }{\kappa^4  }   + \frac{\gamma}{\kappa^2} } +Cd \lambda^2  \\
		\le & H( \widehat A (  I ),   I )+C \lambda ^2 d \sqrt {     \frac{ \lambda^2d}{\kappa^2 }      +     \frac{ \lambda^2 d^2 }{\kappa^4  }   + \frac{\gamma}{\kappa^2} } +C d\lambda^2  .
	\end{align*}
	where the first inequality  follows from \Cref{eq:middle segment}, \Cref{eq:first segment} and \Cref{eq:last segment}, and the second inequality follows from the observation that 
	\[
		\lambda \sqrt { |J_2  | }\| \widehat A(I )\|  \le \lambda \sqrt { | I   | }\| \widehat A(I )\| . 
	\]
\end{proof}

\section{Deviation Bounds}\label{sec-prob-events}
	
\begin{lemma} \label{lemma:Azuma Hoeffding}  
Under all the assumptions in \Cref{theorem:SEPP change point}, for any $t \in [1, T-1]$, let $\epsilon (t) = X(t+1) -  \exp [v +A_m^* (t) g\{\mathcal X ( t ) \}]$.  For any deterministic $v \in \mathbb R^T$ and any integer interval $I \subset [1, T-1]$, it holds that for any $\delta > 0$,
	\[
		\max_{m, m' = 1, \ldots, M} \mathbb{P}\left[\sum_{t \in I }  v_t \epsilon_m(t) g_{m'}\{\mathcal X( t) \} \ge  \delta \right] \le  2\exp\left(-\frac{C\delta ^2 }{\sum_{ t\in I }   v_t^2 }\right).
	\]
\end{lemma}

\begin{proof}
For any $m, m' \in [1, M]$ and $t \in [1, T-1]$, let $Y_t = \sum_{i=1}^t   v_i\epsilon_m (i)g_{m'}\{\mathcal X(i)\}$.  Due to \Cref{model:sepp change point}, we have that $\{Y_t\}$ is a martingale sequence with respect to the filtration $\{\mathcal{F}_t\}$, $\mathcal F_t = \sigma\{X_1,\ldots, X_t\}$.  In addition, with $Y_0 = 0$, for any $t \in [1, T-1]$,
	\[
		\mathbb{E}\left(|Y_t - Y_{t-1}|\right) = \mathbb{E}\left(|v_t \epsilon_m (t) g_{m'}\{\mathcal X(t)\}|\right) \le C|v_t|.
	\]
	The final result follows from Azuma's inequality \citep{azuma1967weighted}.
\end{proof}

\begin{lemma} \label{lemma:finite change points event 1}
Under all the assumptions in \Cref{theorem:SEPP change point}, let $I \subset [1, T]$ be an integer interval and $R \in \mathbb{Z}_+ \cap [1, T]$.  Denote the event 
	\[
		\mathcal A_R (I) = \left\{\max_{m, m' = 1, \ldots, M} \sup_{v: \, \|Dv\|_0  = R, \|v\|_2 = 1} \sum_{t \in I} v_t \epsilon_m(t) g_{m'}\{\mathcal X ( t) \} \ge \sqrt {C\log^2 (TM) R}\right\},
	\]
	for some sufficiently large constant $C$.  It holds that $\mathbb{P}\{\mathcal A_R(I)\} \le (TM)^{-4R}$.
	 
\end{lemma}

\begin{proof} 
In this proof, we use $C$ to refer an absolute constant, which is not necessarily the same throughout the proof.

For any $R \in [1, T] \cap \mathbb{Z}_+$, since $\|Dv\|_0 = R$, $v$ has exactly $R$ change points and there are at most $\binom {T}{R}$ possible choices of the locations of change points.  Given the collection of change points $\{\eta_k\}_{k=1}^R \subset [1, T]$, denote by $\mathcal S(\{\eta_k\}_{k=1}^d)$ the linear subspace of all piecewise-constant functions with all change points at $\{\eta_k\}_{k=1}^R$.  Let $\mathcal N_{1/4}(\{\eta_k\}_{k=1}^R)$ be a $1/4$-net of $\mathcal S(\{\eta_k\}_{k=1}^R)\cap \mathcal{S}(0, 1)$, where $\mathcal{S}(0, 1)$ is the unit sphere in $\mathbb{R}^T$.  Since $\mathcal S(\{\eta_k\}_{k=1}^R)$ is an affine subspace with dimension $R+1$, $\mathcal N_{1/4}(\{\eta_k\}_{k=1}^R)$ can be chosen such that $|\mathcal N_{1/4}(\{\eta_k\}_{k=1}^R)|\le 12^{R+1}$, see e.g.~Lemma 4.1 in \cite{pollard1990empirical} and Lemma 4.2.8 in \cite{vershynin2018high}.
  
Then we have for any fixed $m, m' \in [1, M]$ and any set of $\{\eta_k\}_{k=1}^R \subset [1, T]$, it holds that
	\begin{align*}
		& \mathbb{P}\left[\sup_{v\in \mathcal S (\{\eta_k\}_{k=1}^R),\, \|v\|_2 =1} \sum_{t \in I} v_t \epsilon_m(t) g_{m'}\{\mathcal X (t)\} \ge \sqrt  {C\log^2 (TM) R} \right] \\
		\leq & \mathbb{P}\left[\sup_{v\in \mathcal N_{1/4}( \{\eta_k\}_{k=1}^R)} \sum_{t \in I} v_t \epsilon_m(t) g_{m'}\{\mathcal X (t)\} + \max_{t\in I}|4^{-1}\epsilon_m(t) g_{m'}\{\mathcal X (t)\}| \ge \sqrt  {C\log^2 (TM) R} \right] \\
		\leq & \mathbb{P}\left[\sup_{v\in \mathcal N_{1/4}( \{\eta_k\}_{k=1}^R)} \sum_{t \in I} v_t \epsilon_m(t) g_{m'}\{\mathcal X (t)\} \ge \sqrt  {C\log^2 (TM) R} \right]\\
		& \hspace{6cm} \times \mathbb{P} \left[\max_{t\in I}|4^{-1}\epsilon_m(t) g_{m'}\{\mathcal X (t)\}| < C\log(TM)\right] \\
		& \hspace{4cm} + \mathbb{P} \left[\max_{t\in I}|4^{-1}\epsilon_m(t) g_{m'}\{\mathcal X (t)\}| \geq C\log(TM)\right] \\
		\leq & 12^{R+1} \sup_{v\in \mathcal N_{1/4}( \{\eta_k\}_{k=1}^R)}\mathbb{P}\left[ \sum_{t \in I} v_t \epsilon_m(t) g_{m'}\{\mathcal X (t)\} \ge \sqrt  {C\log^2 (TM) R} \right] \\
		& \hspace{4cm} + \mathbb{P} \left[\max_{t\in I}|4^{-1}\epsilon_m(t) g_{m'}\{\mathcal X (t)\}| \geq C\log(TM)\right]  \\
		\leq & 12^{R+1} \times 2 \exp\left\{-\frac{C\log^2(TM) R}{\sum_{t \in I} v_t^2}\right\} + 2T\exp\left(-\frac{4C\log(TM)}{C_g}\right) \\
		\leq & C\exp\{-C\log(TM) R + R\log(12)\},
	\end{align*}
	where the fourth inequality follows from \Cref{lemma:Azuma Hoeffding} and the sub-Exponential property of Poisson random variables.
	
Therefore for any fixed integer interval $I$ and positive constant $R \leq T$, it holds that
	\[
		\mathbb{P}\{\mathcal{A}_R(I)\} \leq M^2 \binom {T} {R}C\exp\{-C\log(TM) R\} \leq (TM)^{-4R},
	\]
	with sufficiently large absolute constants.
\end{proof}

\begin{lemma} \label{lemma:finite change points event 2}
Let $I \subset [1, T]$ be an integer interval.  Denote the event 
	\[
		\mathcal B (I) = \left\{ \max_{m, m' = 1, \ldots, M} \sum_{t \in I} \epsilon_m(t) g_{m'} \{\mathcal X ( t)\} \ge \sqrt {C|I|\log (TM)}\right\},
	\]
	for some sufficiently large constant $C$.  It holds that $\mathbb{P}\{\mathcal B(I) \} \le (TM)^{-4}$.
\end{lemma}
\begin{proof}
It follows from \Cref{lemma:Azuma Hoeffding} that, with sufficiently large constants,
	\begin{align*}
		& \mathbb{P}\left\{\max_{m, m' = 1, \ldots, M} \sum_{t \in I} \epsilon_m(t) g_{m'}\{\mathcal X ( t)\} \geq \sqrt {C|I|\log (TM)}\right\}	\\
		\leq & M^2 \mathbb{P}\left\{\sum_{t \in I} \epsilon_1(t) g_{2}\{\mathcal X ( t)\} \geq \sqrt {C|I|\log (TM)}\right\} \leq 2M^2 \exp\left(-\frac{C|I| \log(TM)}{2C_g^2 L ^2\sum_{ t\in I }|I|}\right) \leq (TM)^{-4}.
	\end{align*}
\end{proof}



For any $q = 0, \ldots, p-1$, denote 
	\[
		\mathcal T_q = \{pb + q: \, b \in \mathbb Z^+, pb+q = 1, \ldots, T\}.
	\]	
\begin{lemma} \label{lemma:deviation bound 2}
Under all the assumptions in \Cref{theorem:SEPP change point}, let $y(t) = \mathbb{E}[g\{\mathcal X (t)\}g\{\mathcal X (t)\}^\top |\mathcal X (t-p)] - g\{\mathcal X (t)\}g\{\mathcal X (t)\}^\top \in \mathbb R^{M\times M}$, $R \in \mathbb{Z} \cap [1, T]$, $I \subset [1, T]$ be an integer interval and
	\[
		\mathcal C_R (I) = \left\{\max_{m, m' = 1, \ldots, M} \sup_{v:\, \|Dv\|_0 = R, \|v\|_2 = 1} \sum_{t \in I \cap \mathcal T_q} v_t y_{m, m'}(t) \ge \sqrt {C\log (TM)}\right\}.
	\]
	It holds that $\mathbb{P}\{\mathcal C_R (I)\} \le (TM)^{-4}$, with sufficiently large $C > 0$.
\end{lemma}
\begin{proof}
Let $v$ be any fixed vector in  $\mathbb R^T$.  It holds that 
	\[
		\mathbb{E}\{v_t y(t) | \mathcal F_{t-p } \} = 0 \in \mathbb R^{M\times M}. 
	\]
	As a result, $\{Z_t\}_{t \in \mathcal T_q} $ is a martingale sequence with respect to the filtration $\{\mathcal F_{t} \}_{t\in \mathcal T_q}$, where $Z_t = \sum_{s \in \mathcal T_q, \, s \le t}v_s y_{m, m'}(s) $.  In addition, $\mathbb{E}|Z_t - Z_{t -p}| = \mathbb{E}|  v_t  y_{m,m'}  | \le 2 C_g|v_t|$.  By Azuma's inequality \citep{azuma1967weighted}, for any $m, m' \in [1, M]$,
	\[
		\mathbb{P}\left\{\sum_{t \in I}v_t y_{m, m'}(t) \ge \sqrt{C\log(TM)} \right\} \le \exp\left\{-\frac{C\log(TM)}{\sum_{ t\in I }   v_t^2 }\right\}.
	\]
	The desired result follows from  similar covering arguments as in \Cref{lemma:finite change points event 1}.
\end{proof}

\end{document}